\documentclass[11pt]{article}
\usepackage[a4paper, margin=1in]{geometry}

\usepackage{amsmath,amssymb,amsthm}
\usepackage{graphicx}

\usepackage{hyperref}
\usepackage{cite}
\usepackage{enumitem}

\makeatletter
\newtheorem*{rep@theorem}{\rep@title}
\newcommand{\newreptheorem}[2]{%
	\newenvironment{rep#1}[1]{%
		\def\rep@title{#2 \ref{##1}}%
		\begin{rep@theorem}}%
		{\end{rep@theorem}}}
\makeatother

\newenvironment{lemma-repeat}[1]{\begin{trivlist}
		\item[\hspace{\labelsep}{\bf\noindent Lemma \ref{#1} }]\em }%
	{\end{trivlist}}
\newenvironment{theorem-repeat}[1]{\begin{trivlist}
		\item[\hspace{\labelsep}{\bf\noindent Theorem \ref{#1} }]\em }%
	{\end{trivlist}}

\newcommand{\qedsymb}{\qed}
\newenvironment{proofof}[1]{\begin{trivlist}
		\item[\hspace{\labelsep}{\bf\noindent Proof of #1: }]
	}{\qedsymb\end{trivlist}}

\newtheorem{theorem}{Theorem}
\newtheorem{claim}{Claim}
\newtheorem{lemma}{Lemma}

\newtheorem{definition}{Definition}

\def\eps{\varepsilon}

\DeclareMathOperator{\bin}{bin}
\newcommand{\size}[1]{\ensuremath{\left|#1\right|}}
\newcommand{\set}[1]{\left\{ #1 \right\}}

\DeclareMathOperator{\false}{{\scriptstyle{FALSE}}}
\DeclareMathOperator{\true}{{\scriptstyle{TRUE}}}
\def\disj{\mathrm{DISJ}}
\def\eq{\mathrm{EQ}}

\DeclareMathOperator{\wdist}{wdist}
\DeclareMathOperator{\poly}{poly}

\newcommand\singlebar[1]{\bar{#1}}
\newcommand\doublebar[1]{\bar{\bar{#1}}}

\newcommand{\remove}[1]{}

\usepackage{xcolor}

\newcommand{\ifabs}[2]{#1}
\renewcommand{\ifabs}[2]{#2}

\begin{document}

\begin{titlepage}
	
	\title{Quadratic and Near-Quadratic Lower Bounds\\
		for the CONGEST Model\thanks{Department of Computer Science, Technion, \{ckeren,serikhoury,amipaz\}@cs.technion.ac.il. Supported in part by ISF grant 1696/14.}}
	\author{Keren Censor-Hillel \and Seri Khoury \and Ami Paz}
	\maketitle
\maketitle
\begin{abstract}
We present the first super-linear lower bounds for natural graph problems in the CONGEST model, answering a long-standing open question.

Specifically, we show that any exact computation of a minimum vertex cover or a maximum independent set requires $\Omega(n^2/\log^2{n})$ rounds in the worst case in the CONGEST model, as well as any algorithm for $\chi$-coloring a graph, where $\chi$ is the chromatic number of the graph. We further show that such strong lower bounds are not limited to NP-hard problems, by showing two simple graph problems in P which require a quadratic and near-quadratic number of rounds.

Finally, we address the problem of computing an exact solution to weighted all-pairs-shortest-paths (APSP), which arguably may be considered as a candidate for having a super-linear lower bound.
We show a simple $\Omega(n)$ lower bound for this problem,
which implies a separation between the weighted and unweighted cases, since the latter is known to have a complexity of $\Theta(n/\log{n})$. We also formally prove that the standard Alice-Bob framework is incapable of providing a super-linear lower bound for exact weighted APSP, whose complexity remains an intriguing open question.
\end{abstract}
		
	\end{titlepage}
\section{Introduction}

It is well-known and easily proven that many graph problems are \emph{global} for distributed computing, in the sense that solving them necessitates communication throughout the network. This implies tight $\Theta(D)$ complexities, where $D$ is the diameter of the network, for global problems in the LOCAL model. In this model, a message of unbounded size can be sent over each edge in each round, which allows to learn the entire topology in $D$ rounds. Global problems are widely studied in the CONGEST model, in which the size of each message is restricted to $O(\log{n})$ bits, where $n$ is the size of the network. The trivial complexity of learning the entire topology in the CONGEST model is $O(m)$, where $m$ is the number of edges of the communication graph, and since $m$ can be as large as $\Theta(n^2)$, one of the most basic questions for a global problem is how fast in terms of $n$ it can be solved in the CONGEST model.

Some global problems admit fast $O(D)$-round solutions in the CONGEST model, such as constructing a breadth-first search tree~\cite{Peleg:book00}. Some others have complexities of $\tilde{\Theta}(D+\sqrt{n})$, such as constructing a minimum spanning tree, and various approximation and verification problems~\cite{FrischknechtHW12,SarmaHKKNPPW12,KorKP11,PelegR00,PelegRT12,HolzerPRW14}. Some problems are yet harder, with complexities that are near-linear in $n$~\cite{PelegRT12,AbboudCHK16,LenzenP13,FrischknechtHW12,HolzerW12}. For some problems, no $O(n)$ solutions are known and they are candidates to being even harder that the ones with linear-in-$n$ complexities.

A major open question about global graph problems in the CONGEST model is whether natural graph problems for which a super-linear number of rounds is required indeed exist. In this paper, we answer this question in the affirmative. That is, our conceptual contribution is that \textbf{there exist super-linearly hard problems in the CONGEST model}. In fact, the lower bounds that we prove in this paper are as high as quadratic in $n$, or quadratic up to logarithmic factors, and hold even for networks of a constant diameter. Our lower bounds also imply linear and near-linear lower bounds for the CLIQUE-BROADCAST model.

We note that high lower bounds for the CONGEST model may be obtained rather artificially, by forcing large inputs and outputs that must be exchanged. However, we emphasize that all the problems for which we show our lower bounds can be reduced to simple decision problems, where each node needs to output a single bit. All inputs to the nodes, if any, consist of edge weights that can be represented by $\mbox{polylog} n$ bits.

Technically, we prove a lower bound of $\Omega(n^2/\log^2{n})$ on the number of rounds required for computing an exact minimum vertex cover, which also extends to computing an exact maximum independent set (Section~\ref{sec:mvc}).
This is in stark contrast to the recent $O(\log\Delta/\log\log\Delta)$-round algorithm of~\cite{Bar-YehudaCS16} for obtaining a $(2+\epsilon)$-approximation to the minimum vertex cover.
Similarly, we give an $\Omega(n^2/\log^2{n})$ lower bound for $3$-coloring a $3$-colorable graph, which extends also for deciding whether a graph is $3$-colorable, and also implies the same hardness for computing the chromatic number $\chi$ or computing a $\chi$-coloring (Section~\ref{sec:coloring}). These lower bounds hold even for randomized algorithms which succeed with high probability.\footnote{We say that an event occurs with high probability (w.h.p) if it occurs with probability $\frac{1}{n^c}$, for some constant $c>0$.}

An immediate question that arises is whether only NP-hard problems are super-linearly hard in the CONGEST model.
In Section~\ref{sec:P}, we provide a negative answer to such a postulate, by showing two simple problems that admit polynomial-time sequential algorithms, but in the CONGEST model require $\Omega(n^2)$ rounds (identical subgraph detection) or $\Omega(n^2/\log{n})$ rounds (weighted cycle detection). The latter also holds for  randomized algorithms, while for the former we show a randomized algorithm that completes in $O(D)$ rounds, providing the strongest possible separation between deterministic and randomized complexities for global problems in the CONGEST model.

Finally, we address the intriguing open question of the complexity of computing exact weighted all-pairs-shortest-paths (APSP) in the CONGEST model.
While the complexity of the unweighted version of APSP is $\Theta(n/\log{n})$, as follows from~\cite{FrischknechtHW12,HuaFQALSJ16}, the complexity of weighted APSP remains largely open, and only recently the first sub-quadratic algorithm was given in~\cite{Elkin17}. With the current state-of-the-art, this problem could be considered as a suspect for having a super-linear complexity in the CONGEST model.
While we do not pin-down the complexity of weighted APSP in the CONGEST model, we provide a truly linear lower bound of $\Omega(n)$ rounds for it, which separates its complexity from that of the unweighted case. Moreover, we argue that it is not a coincidence that we are currently unable to show super-linear lower bound for weighted APSP, by formally proving that the commonly used framework of reducing a $2$-party communication problem to a problem in the CONGEST model cannot provide a super-linear lower bound for weighted APSP, regardless of the function and the graph construction used (Section~\ref{sec:APSP}). This implies that \textbf{obtaining any super-linear lower bound for weighted APSP provably requires a new technique}.

\subsection{The Challenge}

Many lower bounds for the CONGEST model rely on reductions from $2$-party communication problems (see, e.g.,~\cite{AbboudCHK16,FrischknechtHW12,HolzerW12,Nanongkai14,PelegR00,SarmaHKKNPPW12,DruckerKO13,Elkin06,NanongkaiSP11,Censor-HillelKP16}). In this setting, two players, Alice and Bob, are given inputs of $K$ bits and need to a single output a bit according to some given function of their inputs. One of the most common problem for reduction is Set Disjointness, in which the players need to decide whether there is an index for which both inputs are $1$. That is, if the inputs represent subsets of $\set{0,\ldots,K-1}$, the output bit of the players needs to indicate whether their input sets are disjoint. The communication complexity of $2$-party Set Disjointness is known to be $\Theta(K)$~\cite{KushilevitzN:book96}.

In a nutshell, there are roughly two standard frameworks for reducing the $2$-party communication problem of computing a function $f$ to a problem $P$ in the CONGEST model. One of these frameworks works as follows. A graph construction is given, which consists of some fixed edges and some edges whose existence depends on the inputs of Alice and Bob. This graph should have the property that a solution to $P$ over it determines the solution to $f$. Then, given an algorithm $ALG$ for solving $P$ in the CONGEST model, the vertices of the graph are split into two disjoint sets, $V_A$ and $V_B$, and Alice simulates $ALG$ over $V_A$ while Bob simulates $ALG$ over $V_B$. The only communication required between Alice and Bob in order to carry out this simulation is the content of messages sent in each direction over the edges of the cut $C=E(V_A,V_B)$.
Therefore, given a graph construction with a cut of size $|C|$ and inputs of size $K$ for a function $f$ whose communication complexity on $K$ bits is at least $CC(f)$, the round complexity of $ALG$ is at least $\Omega(CC(f)/|C|\log{n})$.

The challenge in obtaining super-linear lower bounds was previously that the cuts in the graph constructions were large compared with the input size $K$. For example, the graph construction for the lower bound for computing the diameter in~\cite{FrischknechtHW12} has $K=\Theta(n^2)$ and $|C|=\Theta(n)$, which gives an almost linear lower bound. The graph construction in~\cite{FrischknechtHW12} for the lower bound for computing a ($3/2-\epsilon$)-approximation to the diameter has a smaller cut of $|C|=\Theta(\sqrt{n})$, but this comes at the price of supporting a smaller input size $K=\Theta(n)$, which gives a lower bound that is roughly a square-root of $n$.

To overcome this difficulty, we leverage the recent framework of~\cite{AbboudCHK16}, which provides a bit-gadget whose power is in allowing a logarithmic-size cut. We manage to provide a graph construction that supports inputs of size $K=\Theta(n^2)$ in order to obtain our lower bounds for minimum vertex cover, maximum independent set and $3$-coloring\footnote{It can also be shown, by simple modifications to our constructions, that these problems require $\Omega(m)$ rounds, for graphs with $m$ edges.}. The latter is also inspired by, and is a simplification of, a lower bound construction for the size of proof labelling schemes~\cite{GoosS16}.

Further, for the problems in P that we address, the cut is as small as $|C|=O(1)$. For one of the problems, the size of the input is such that it allows us to obtain the highest possible lower bound of $\Omega(n^2)$ rounds.

With respect to the complexity of the weighted APSP problem, we show an embarrassingly simple graph construction that extends a construction of~\cite{Nanongkai14}, which leads to an $\Omega(n)$ lower bound. However, we argue that a new technique must be developed in order to obtain any super-linear lower bound for weighted APSP. Roughly speaking, this is because given a construction with a set $S$ of nodes that touch the cut, Alice and Bob can exchange $O(|S|n\log{n})$ bits which encode the weights of all lightest paths from any node in their set to a node in $S$.
Since the cut has $\Omega(|S|)$ edges,
and the bandwidth is $\Theta(\log{n})$,
this cannot give a lower bound of more than $\Omega(n)$ rounds. With some additional work, our proof can be carried over to a larger number of players at the price of a small logarithmic factor, as well as to the second Alice-Bob framework used in previous work (e.g.~\cite{SarmaHKKNPPW12}), in which Alice and Bob do not simulate nodes in a fixed partition, but rather in decreasing sets that partially overlap.
Thus, determining the complexity of weighted APSP requires new tools, which we leave as a major open problem.

\subsection{Additional Related Work}

\noindent\textbf{Vertex Coloring, Minimum Vertex Cover, and Maximum Independent Set:} One of the most central problems in graph theory is vertex coloring, which has been extensively studied in the context of distributed computing (see, e.g.,~\cite{BarenboimEPS16,Barenboim16,BarenboimE11,BarenboimE14,BarenboimEK14,Linial92,EmekPSW14,FraigniaudGIP09,FraigniaudHK16,HarrisSS16,MoscibrodaW08,SchneiderW11,PettieS15,ChungPS14,ChangKP16,ColeV86,Barenboim12} and references therein). The special case of finding a $(\Delta +1)$-coloring, where $\Delta$ is the maximum degree of a node in the network, has been the focus of many of these studies, but is a \emph{local} problem, which can be solved in much less than a sublinear number of rounds.

Another classical problem in graph theory is finding a minimum vertex cover (MVC). In distributed computing, the time complexity of approximating MVC has been addressed in several cornerstone studies~\cite{AstrandFPRSU09,Bar-YehudaCS16,AstrandS10,GrandoniKP08,GrandoniKPS08,KhullerVY94,KoufogiannakisY09,KuhnMW16,PolishchukS09,BarenboimEPS16,HanckowiakKP01,PanconesiR01,KuhnMW06}.

Observe that finding a minimum size vertex cover is equivalent to finding a maximum size independent set. However, these problems are not equivalent in an approximation-preserving way. Distributed approximations for maximum independent set has been studied in~\cite{LenzenW08,CzygrinowHW08,BodlaenderHKK16,BYCHGS17}.

~\\
\noindent\textbf{Distance Computations:} Distance computation problems have been widely studied in the CONGEST model for both weighted and unweighted networks~\cite{AbboudCHK16,FrischknechtHW12,HolzerW12,HolzerPRW14,PelegRT12,HolzerP14,LenzenP15,LenzenP13,Nanongkai14,HuaFQALSJ16,HenzingerKN16}. One of the most fundamental problems of distance computations is computing all pairs shortest paths. For unweighted networks, an upper bound of $O(n/\log n)$ was recently shown by~\cite{HuaFQALSJ16},
matching the lower bound of~\cite{FrischknechtHW12}.
Moreover, the possibility of bypassing this near-linear barrier for any constant approximation factor was ruled out by~\cite{Nanongkai14}. For the weighted case, however, we are still very far from understanding the complexity of APSP, as there is still a huge gap between the upper and lower bounds. Recently, Elkin~\cite{Elkin17} showed an $O(n^{\frac{5}{3}}\cdot\log^{\frac{2}{3}}(n))$ upper bound for weighted APSP, while the previously highest lower bound was the near-linear lower bound of~\cite{Nanongkai14} (which holds also for any $(\poly n)$-approximation factor in the weighted case).

Distance computation problems have also been considered in the CONGESTED-CLIQUE model~\cite{HenzingerKN16,CensorKKLPS15,HolzerP14}, in which the underlying communication network forms a clique. In this model~\cite{CensorKKLPS15} showed that
unweighted APSP, and a $(1+o(1))$-approximation for weighted APSP,
can be computed in $O(n^{0.158})$ rounds.

~\\
\noindent\textbf{Subgraph Detection:}
The problem of finding subgraphs of a certain topology has received a lot of attention in both the sequential and the distributed settings (see, e.g.,~\cite{AbboudL13,WilliamsW13,DahlgaardKS17,JorgensenP14,MarxP14,AbboudLW14,DruckerKO13,DolevLP12,AlonYZ97,CensorKKLPS15} and references therein).
The problems of finding paths of length 4 or 5 with zero weight are also related to other fundamental problems, notable in our context is APSP~\cite{AbboudL13}.

\section{Preliminaries}
\label{sec:preliminaries}
\subsection{Communication Complexity}

In a two-party communication complexity problem~\cite{KushilevitzN:book96}, there is a function $f:\set{0,1}^K\times\set{0,1}^K\to\set{\true,\false}$, and two players, Alice and Bob, who are given two input strings, $x,y\in\set{0,1}^K$, respectively, that need to compute $f(x,y)$. The \emph{communication complexity} of a protocol $\pi$ for computing $f$, denoted $CC(\pi)$, is the maximal number of bits Alice and Bob exchange in $\pi$, taken over all values of the pair $(x,y)$. The \emph{deterministic communication complexity} of $f$, denoted $CC(f)$, is the minimum over $CC(\pi)$, taken over all deterministic protocols $\pi$ that compute $f$.

In a \emph{randomized protocol} $\pi$, Alice and Bob may each use a random bit string. A randomized protocol $\pi$ computes $f$ if the probability, over all possible bit strings, that $\pi$ outputs $f(x,y)$ is at least $2/3$. The \emph{randomized communication complexity} of $f$, $CC^R(f)$, is the minimum over $CC(\pi)$, taken over all randomized protocols $\pi$ that compute $f$.

In the \emph{Set Disjointness} problem ($\disj_K$), the function $f$ is $\disj_K(x,y)$, whose output is $\false$ if there is an index $i\in\set{0,\ldots,K-1}$ such that $x_i=y_i=1$, and $\true$ otherwise. In the \emph{Equality} problem ($\eq_K$), the function $f$ is  $\eq_K(x,y)$, whose output is $\true$ if $x=y$, and $\false$ otherwise.

Both the deterministic and randomized communication complexities of the $\disj_K$ problem are known to be $\Omega(K)$~\cite[Example 3.22]{KushilevitzN:book96}. The deterministic communication complexity of $\eq_K$ is in $\Omega(K)$~\cite[Example 1.21]{KushilevitzN:book96}, while its randomized communication complexity is in $\Theta(\log K)$~\cite[Example 3.9]{KushilevitzN:book96}.

\subsection{Lower Bound Graphs}
To prove lower bounds on the number of rounds necessary in order to solve a distributed problem in the CONGEST model, we use reductions from two-party communication complexity problems. To formalize them we use the following definition.

\begin{definition}(Family of Lower Bound Graphs)\newline
\label{def:family}
Fix an integer $K$, a function $f:\set{0,1}^K\times\set{0,1}^K\to\set{\true,\false}$ and a predicate $P$ for graphs. The family of graphs $\{G_{x,y}=(V,E_{x,y})\mid x,y\in\set{0,1}^K\}$, is said to be a family of \emph{lower bound graphs w.r.t. $f$ and $P$} if the following properties hold:
\begin{enumerate}
  \item[(1)] The set of nodes $V$ is the same for all graphs, and we denote by $V=V_A\dot\cup V_B$ a fixed partition of it;
  \item[(2)] Only the existence or the weight of edges in $V_A\times V_A$ may depend on $x$;
  \item[(3)] Only the existence or the weight of edges in $V_B\times V_B$ may depend on $y$;
  \item[(4)] $G_{x,y}$ satisfies the predicate $P$ iff $f(x,y)=\true$.
\end{enumerate}
\end{definition}
We use the following theorem, which is standard in the context of communication complexity-based lower bounds for the CONGEST model (see, e.g.~\cite{AbboudCHK16,FrischknechtHW12,DruckerKO13,HolzerP14}) Its proof is by a standard simulation argument.
\begin{theorem}
\label{thm: general lb framework}
Fix a function $f:\set{0,1}^K\times\set{0,1}^K\to\set{\true,\false}$ and a predicate $P$. If there is a family $\{G_{x,y}\}$ of lower bound graphs with $C = E(V_A, V_B)$ then any deterministic algorithm for deciding $P$ in the CONGEST model requires $\Omega (CC(f)/\size{C}\log n)$ rounds, and any randomized algorithm for deciding $P$ in the CONGEST model requires $\Omega (CC^R(f)/\size{C}\log n)$ rounds.
\end{theorem}

\begin{proof}
Let $ALG$ be a distributed algorithm in the CONGEST model that decides $P$ in $T$ rounds. Given inputs $x,y \in \set{0,1}^K$ to Alice and Bob, respectively, Alice constructs the part of $G_{x,y}$ for the nodes in $V_A$ and Bob does so for the nodes in $V_B$. This can be done by items (1),(2) and (3) in Definition~\ref{def:family}, and since $\set{G_{x,y}}$ satisfies this definition. Alice and Bob simulate $ALG$ by exchanging the messages that are sent during the algorithm between nodes of $V_A$ and nodes of $V_B$ in either direction. (The messages within each set of nodes are simulated locally by the corresponding player without any communication). Since item (4) in Definition~\ref{def:family} also holds, we have that Alice and Bob correctly output $f(x,y)$ based on the output of $ALG$. For each edge in the cut, Alice and Bob exchange $O(\log{n})$ bits per round. Since there are $T$ rounds and $\size{C}$ edges in the cut, the number of bits exchanged in this protocol for computing $f$ is $O(T\size{C}\log{n})$. The lower bounds for $T$ now follows directly from the lower bounds for $CC(f)$ and $CC^R(f)$.
\end{proof}

In what follows, for each decision problem addressed, we describe a fixed graph construction $G=(V,E)$, which we then generalize to a family of graphs $\{G_{x,y}=(V,E_{x,y})\mid x,y\in\set{0,1}^K\}$, which we show to be a family lower bound graphs w.r.t. to some function $f$ and the required predicate $P$. By Theorem ~\ref{thm: general lb framework} and the known lower bounds for the $2$-party communication problem, we deduce a lower bound for any algorithm for deciding $P$ in the CONGEST model.

~\\
\textbf{Remark:} For our constructions which use the Set Disjointness function as $f$, we need to exclude the possibilities of all-$1$ input vectors. This is for the sake of guaranteeing that the graphs are connected, in order to avoid trivial impossibilities. However, this restriction does not change the asymptotic bounds for Set Disjointness, since computing this function while excluding all-$1$ input vectors can be reduced to computing this function for inputs that are shorter by one bit (by having the last bit be fixed to $0$).

\ifabs{
\section{Near-Quadratic Lower Bound for Minimum Vertex Cover}
\label{sec:nphard}
\label{sec:mvc}
}{
\section{Near-Quadratic Lower Bounds for NP-Hard Problems}
\label{sec:nphard}
\subsection{Minimum Vertex Cover}
\label{sec:mvc}
}
The first near-quadratic lower bound we present is for computing a minimum vertex cover, as stated in the following theorem.
\begin{theorem}\label{thm:VC}
	Any distributed algorithm in the CONGEST model for computing a minimum vertex cover or for deciding whether there is a vertex cover of a given size $M$ requires $\Omega(n^2/\log^2n)$ rounds.
\end{theorem}

Finding the minimum size of a vertex cover is equivalent to finding the maximum size of a maximum independent set, because a set of nodes is a vertex cover if and only if its complement is an independent set.
Thus, Theorem~\ref{thm:MaxIS} is a direct corollary of Theorem~\ref{thm:VC}.

\begin{theorem}\label{thm:MaxIS}
	Any distributed algorithm in the CONGEST model for computing a maximum independent set or for deciding whether there is an independent set of a given size
requires $\Omega(n^2/\log^2n)$ rounds.
\end{theorem}

Observe that a lower bound for deciding whether there is a vertex cover of some given size $M$ or not implies a lower bound for computing a minimum vertex cover. This is because computing the size of a given subset of nodes can be easily done in $O(D)$ rounds using standard tools. Therefore, to prove Theorem~\ref{thm:VC} it is sufficient to prove its second part. We do so by describing a family of lower bound graphs with respect to the Set Disjointness function and the predicate $P$ that says that the graph has a vertex cover of size $M$. We begin with describing the fixed graph construction $G=(V,E)$ and then define the family of lower bound graphs and analyze its relevant properties.

~\\
\noindent\textbf{The fixed graph construction:}
\begin{figure}[t]
	\begin{center}
		\includegraphics[scale=0.45]{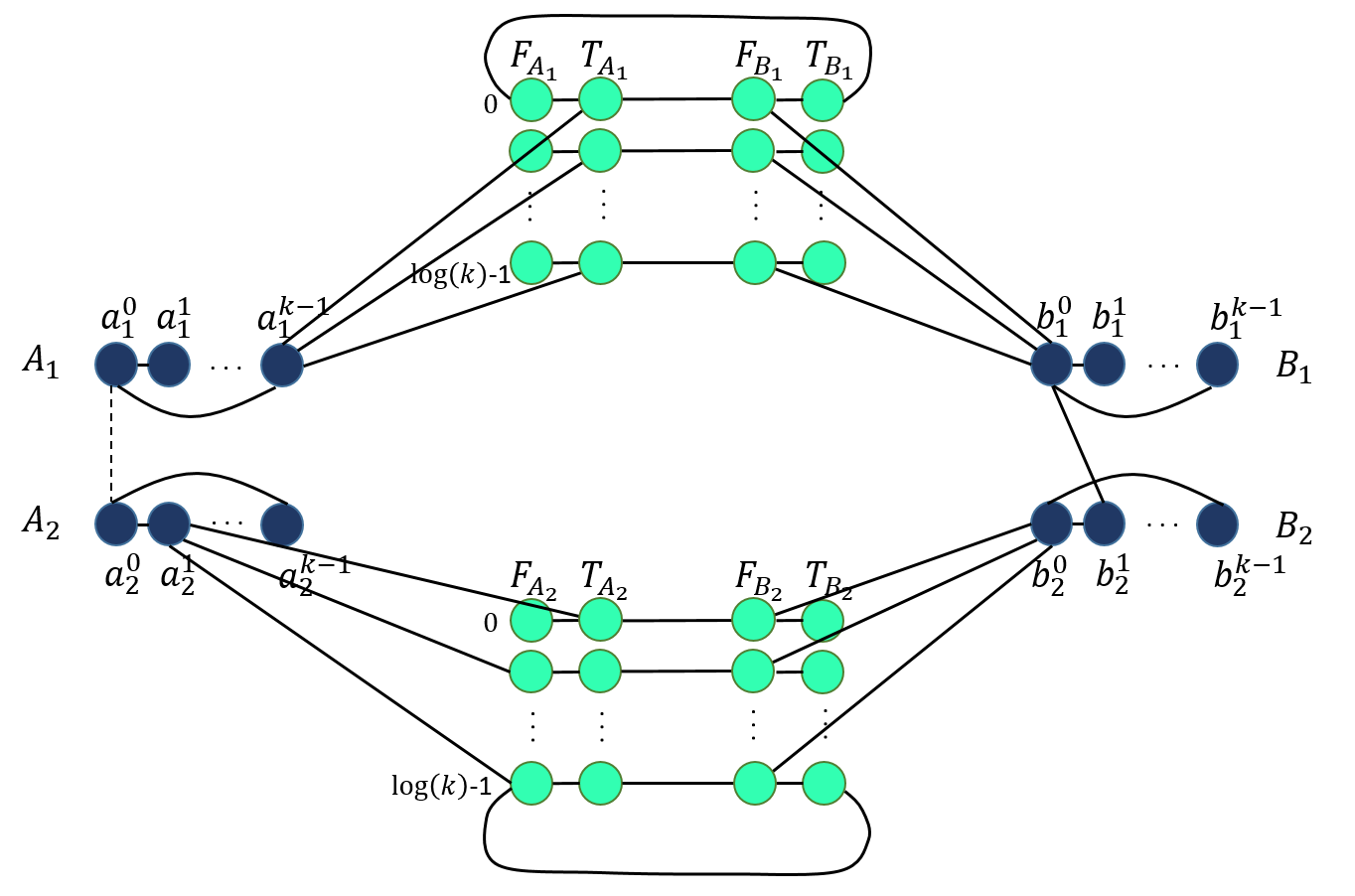}
	\end{center}
	\caption{{\small The family of lower bound graphs for deciding the size of a vertex cover, with many edges omitted for clarity. The node $a^{k-1}_1$ is connected to all the nodes in $T_{A_1}$, and $a^1_2$ is connected to $t^0_{A_2}$ and to all the nodes in $F_{A_2}\setminus \{f^0_{A_2}\}$. Examples of edges from  $b^0_1$ and $b^{0}_2$ to the bit-gadgets are also given. An additional edge, which is among the edges corresponding to the strings $x$ and $y$, is $\{b^0_1, b^1_2\}$, while the edge $\{a^0_1, a^0_2\}$ does not exist. Here, $x_{0,0}=1$ and $y_{0,1}=0$.}}
	\label{fig: consVC}
\end{figure}
Let $k$ be a power of $2$. The fixed graph (Figure~\ref{fig: consVC}) consists of four cliques of size $k$: $A_1=\{a^i_1\mid 0\leq i\leq k-1\}$, $A_2=\{a^i_2\mid 0\leq i\leq k-1\}$, $B_1=\{b^i_1\mid 0\leq i\leq k-1\}$ and $B_2=\{b^i_2\mid 0\leq i\leq k-1\}$.
In addition, for each set $S\in \set{A_1,A_2,B_1,B_2}$,
there are two corresponding sets of nodes of size $\log k$, denoted $F_{S}=\{f^h_{S}\mid 0\leq h\leq \log k-1\}$ and $T_{S}=\{t^h_{S}\mid 0\leq h\leq \log k -1\}$.
The latter are called \emph{bit-gadgets} and their nodes are \emph{bit-nodes}.

The bit-nodes are partitioned into $2\log k$ $4$-cycles: for each $h\in \set{0,\ldots,\log k-1}$ and $\ell\in\set{1,2}$,
we connect the $4$-cycle $(f^h_{A_\ell},t^h_{A_\ell},f^h_{B_\ell},t^h_{B_\ell})$.
Note that there are no edges between pairs of nodes denoted $f^h_S$,
or between pairs of nodes denoted  $t^h_S$.

The nodes of each set $S\in \set{A_1,A_2,B_1,B_2}$ are connected to nodes in the corresponding set of bit-nodes, according to their binary representation, as follows.
Let $s^i_\ell$ be a node in a set $S\in \{A_1,A_2,B_1,B_2\}$, i.e.\ $s\in \set{a,b}$, $\ell\in\set{1,2}$ and $i\in\set{0,\ldots,k-1}$, and let $i_h$ denote the bit number $h$ in the binary representation of $i$.
For such a node $s^i_\ell$
define $\bin(s^i_\ell) = \set{f^h_S\mid i_h=0}\cup \set{t^h_S\mid i_h=1}$,
and connect $s^i_\ell$ by an edge to each of the nodes in $\bin(s^i_\ell)$.
The next two claims address the basic properties of vertex covers of $G$.

\begin{claim}\label{ObSize}
Any vertex cover of $G$ must contain at least $k-1$ nodes from each of the clique $A_1,A_2,B_1$ and $B_2$,
and at least $4\log k$ bit-nodes.
\end{claim}

\begin{proof}
In order to cover all the edges of each if the cliques on $A_1,A_2,B_1$ and $B_2$, any vertex cover must contain at least $k-1$ nodes of the clique.
For each $h\in\set{0,\ldots, \log k -1}$ and $\ell\in \{1,2\}$,
in order to cover the edges of the 4-cycle $(f^h_{A_\ell},t^h_{A_\ell},f^h_{B_\ell},t^h_{B_\ell})$,
any vertex cover must contain at least two of the cycle nodes.
\end{proof}
	
\begin{claim}
\label{claim:VC relating AB}
If $U\subseteq V$ is a vertex cover of $G$ of size $4(k-1)+4\log k$,
then there are two indices $i,j\in\set{0,\ldots,k-1}$
such that $a^i_1,a^j_2,b^i_1,b^j_2$ are not in $U$.
\end{claim}

\begin{proof}
By Claim~\ref{ObSize},
$U$ must contain $k-1$ nodes from each clique $A_1,A_2,B_1$ and $B_2$, and $4\log k$ bit-nodes,
so it must not contain one node from each clique.
Let $a^i_1,a^j_2,b^{i'}_1,b^{j'}_2$ be the nodes in $A_1, A_2,B_1,B_2$ which are not in $U$, respectively.
To cover the edges connecting $a^i_1$ to $\bin(a^i_1)$, $U$ must contain all the nodes of $\bin(a^i_1)$, and similarly, $U$ must contain all the nodes of $\bin(b^{i'}_1)$.
If $i\neq i'$ then there is an index $h\in\set{0,\ldots,\log k-1}$ such that $i_h\neq i'_h$,
so one of the edges $(f^h_{A_1},t^h_{B_1})$ or $(t^h_{A_1},f^h_{B_1})$ is not covered by $U$.
Thus, it must hold that $i=i'$.
A similar argument shows $j=j'$.
\end{proof}

\noindent\textbf{Adding edges corresponding to the strings $x$ and $y$:}
Given two binary strings $x,y\in\set{0,1}^{k^2}$,
we augment the graph $G$ defined above
with additional edges, which defines $G_{x,y}$.
Assume that $x$ and $y$ are indexed by pairs of the form $(i,j)\in \set{0,\ldots,k-1}^2$.
For each such pair $(i,j)$
we add to $G_{x,y}$ the following edges.
If $x_{i,j}=0$, then we add an edge between the nodes $a_1^i$ and $a_2^j$, and if $y_{i,j}=0$ then we add an edge between the nodes $b_1^i$ and $b_2^j$.
To prove that $\set{G_{xy}}$ is a family of lower bound graphs, it remains to prove the next lemma.

\begin{lemma}\label{mainLemmaVC}
The graph $G_{x,y}$ has a vertex cover of cardinality $M=4(k-1)+4\log k$
iff $\disj(x,y)=\false$.
\end{lemma}

\begin{proof}
For the first implication, assume that $\disj(x,y)=\false$
and let $i,j\in\set{0,\ldots,k-1}$ be such that $x_{i,j}=y_{i,j}=1$.
Note that in this case $a^i_1$ is not connected to $a^j_2$,
and $b^i_1$ is not connected to $b^j_2$.
We define a set $U \subseteq V$ as the union of the two sets of nodes $(A_1\setminus \{a^i_1\})\cup (A_2\setminus \{a^j_2\})\cup (B_1\setminus \{b^i_1\})\cup (B_2\setminus \{b^j_2\})$ and $\bin(a^i_1)\cup \bin(a^j_2)\cup\bin(b^i_1)\cup\bin(b^j_2)$,
and show that $U$ is a vertex cover of $G_{x,y}$.

First, $U$ covers all the edges inside the cliques $A_1,A_2,B_1$ and $B_2$,
as it contains $k-1$ nodes from each clique.
These nodes also cover all the edges connecting nodes in $A_1$ to nodes in $A_2$ and all the edges connecting nodes in $B_1$ to nodes in $B_2$.
Furthermore, $U$ covers any edge connecting some node
$u\in (A_1\setminus \{a^i_1\})\cup (A_2\setminus \{a^j_2\})\cup (B_1\setminus \{b^i_1\})\cup (B_2\setminus \{b^j_2\})$ with the bit-gadgets.
For each node $s\in a^i_1,a^j_2,b^i_1,b^j_2$,
the nodes $\bin(s)$ are in $U$,
so $U$ also cover the edges connecting $s$ to the bit gadget.
Finally, $U$ covers all the edges inside the bit gadgets,
as from each $4$-cycle
$(f^h_{A_\ell},t^h_{A_\ell},f^h_{B_\ell},t^h_{B_\ell})$
it contains two non-adjacent nodes:
if $i_h=0$ then $f^h_{A_1},f^h_{B_1} \in U$
and otherwise $t^h_{A_1},t^h_{B_1} \in U$,
and
if $j_h=0$ then $f^h_{A_2},f^h_{B_2} \in U$
and otherwise $t^h_{A_2},t^h_{B_2} \in U$.
We thus have that $U$ is a vertex cover of size $4(k-1)+4\log k$,
as needed.
	
For the other implication, let $U\subseteq V$ be a vertex cover of $G_{x,y}$ of size $4(k-1)+4\log k$.
As the set of edges of $G$ is contained in the set of edges of $G_{x,y}$,
$U$ is also a cover of $G$,
and by Claim~\ref{claim:VC relating AB}
there are indices $i,j\in\set{0,\ldots,k-1}$
such that $a^i_1,a^j_2,b^i_1,b^j_2$ are not in $U$.
Since $U$ is a cover, the graph does not contain the edges
$(a^i_1, a^j_2)$ and $(b^i_1, b^j_2)$,
so we conclude that $x_{i,j}=y_{i,j} =1$,
which implies that $\disj(x,y)=\false$.
\end{proof}

Having constructed the family of lower bound graphs, we are now ready to prove Theorem~\ref{thm:VC}.
\begin{proofof}{Theorem~\ref{thm:VC}} To complete the proof of Theorem~\ref{thm:VC}, we divide the nodes of $G$ (which are also the nodes of $G_{x,y}$) into two sets. Let $V_A = A_1 \cup A_2 \cup F_{A_1} \cup T_{A_1} \cup F_{A_2} \cup T_{A_2}$ and $V_B=V\setminus V_A$.
Note that $n\in \Theta(k)$, and thus $K=|x|=|y|=\Theta(n^2)$.
Furthermore, note that the only edges in the cut $E(V_A, V_B)$
are the edges between nodes in
$\{F_{A_1} \cup T_{A_1}\cup F_{A_2} \cup T_{A_2}\}$ and nodes in $\{F_{B_1} \cup T_{B_1}\cup F_{B_2} \cup T_{B_2}\}$,
which are in total $\Theta(\log n)$ edges. Since Lemma~\ref{mainLemmaVC} shows that $\{G_{x,y}\}$ is a family of lower bound graphs, we can apply Theorem~\ref{thm: general lb framework} on the above partition to deduce that because of the lower bound for Set Disjointness, any algorithm in the CONGEST model for deciding whether a given graph has a cover of cardinality $M=4(k-1)+4\log k$ requires at least $\Omega(K/\log^2(n))=\Omega(n^2/\log^2(n))$ rounds.
\end{proofof}

\newcommand{\ColoringSection}{

Given a graph $G$, we denote by $\chi(G)$ the minimal number of colors in a proper vertex-coloring of $G$.
In this section we consider the problems of
coloring a graph in $\chi$ colors,
computing $\chi$
and approximating it.
We prove the next theorem.

\begin{theorem}
\label{thm: coloring veriants lb}
Any distributed algorithm in the CONGEST model
that colors a $\chi$-colorable graph $G$ in $\chi$ colors
or compute $\chi(G)$
requires $\Omega(n^2/\log^2n)$ rounds.

Any distributed algorithm in the CONGEST model
that decides if $\chi(G)\leq c$ for a given integer $c$,
requires $\Omega((n-c)^2/(c\log n+\log^2n))$ rounds.
\end{theorem}

We give a detailed lower bound construction for the first part of the theorem,
by showing that distinguishing $\chi\leq3$ from $\chi\geq4$
is hard.
Then, we extend our construction to deal with deciding whether $\chi\leq c$.

~\\
\noindent\textbf{The fixed graph construction:}  We describe a family of lower bound graphs, which builds upon the family of graphs defined in Section~\ref{sec:mvc}. We define $G=(V,E)$ as follows (see Figure~\ref{fig:coloring}).
\begin{figure}[t]
	\begin{center}
		\includegraphics[scale=0.5]{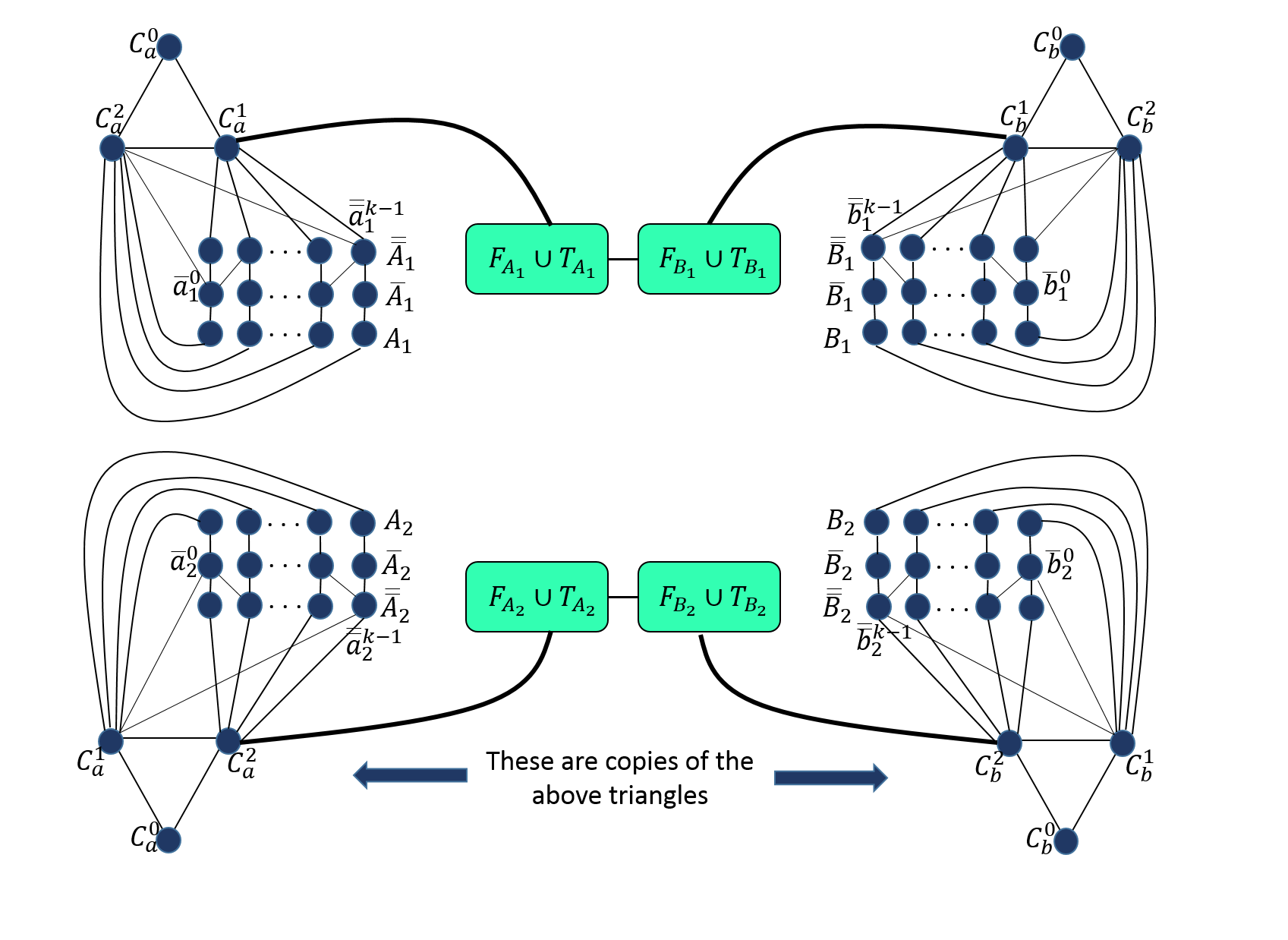}
	\end{center}
	\caption{{\small The family of lower bound graphs for coloring, with many edges omitted for clarity. The node $C^1_a$ is connected to all the nodes in $F_{A_1}\cup T_{A_1}$ and $C^1_b$ is connected to all the nodes in $F_{B_1}\cup T_{B_1}$. The node $C^2_a$ is connected to all the nodes in $F_{A_2}\cup T_{A_2}$ and $C^2_b$ is connected to all the nodes in $F_{B_2}\cup T_{B_2}$.}}
	\label{fig:coloring}
\end{figure}

There are four sets of size $k$: $A_1=\{a^i_1\mid 0\leq i\leq k-1\}$, $A_2=\{a^i_2\mid 0\leq i\leq k-1\}$, $B_1=\{b^i_1\mid 0\leq i\leq k-1\}$ and $B_2=\{b^i_2\mid 0\leq i\leq k-1\}$. As opposed to the construction in Section~\ref{sec:mvc}, the nodes of these sets are not connected to one another. In addition, as in Section~\ref{sec:mvc}, for each set $S\in \set{A_1,A_2,B_1,B_2}$, there are two corresponding sets of nodes of size $\log k $, denoted $F_{S}=\{f^h_{S}\mid 0\leq h\leq \log k -1\}$ and $T_{S}=\{t^h_{S}\mid 0\leq h\leq \log k -1\}$.
For each $h\in \set{0,\ldots,\log k -1}$ and $\ell\in\set{1,2}$, the nodes $(f^h_{A_\ell},t^h_{A_\ell},f^h_{B_\ell},t^h_{B_\ell})$ constitute a $4$-cycle.
Each node $s^i_\ell$ in a set $S\in \set{A_1,A_2,B_1,B_2}$ is connected by to all nodes in $\bin(s^i_\ell)$.
Up to here, the construction differs from the construction in Section~\ref{sec:mvc} only by not having edges inside the sets $A_1,A_2,B_1,B_2$.

We now add the following two gadgets to the graph.
\begin{enumerate}
    \item We add three nodes $C^0_a, C^1_a, C^2_a$ connected as a triangle, another set of three nodes $C^0_b, C^1_b, C^2_b$ connected as a triangle, and edges connecting $C^i_a$ to $C^j_b$ for each $i\neq j \in\set{0,1,2}$. We connect all the nodes of the form $f^h_{A_1},t^h_{A_1}$, $h\in \set{0,\ldots, \log k -1}$, to $C_a^1$. Similarly, we connect all the nodes $f^h_{B_1},t^h_{B_1}$ to $C_b^1$, the nodes $f^h_{A_2},t^h_{A_2}$ to $C_a^2$ and the nodes $f^h_{B_2},t^h_{B_2}$ to $C_b^2$.
    \item For each set $S \in A_1,A_2,B_1,B_2$, we add two sets of nodes, $\singlebar S = \set{\singlebar s^i_{\ell}\mid s^i_\ell \in S}$ and $\doublebar S = \set{\doublebar s^i_{\ell}\mid s^i_\ell \in S}$.
        For each $\ell\in\set{1,2}$ and $i\in\set{0,\ldots,k-1}$ we connect a path $(s^i_{\ell}, \singlebar s^i_{\ell}, \doublebar s^i_{\ell})$, and for each $\ell\in\set{1,2}$ and $i\in\set{0,\ldots,n-2}$, we connect $\doublebar s^i_{\ell}$ to $\singlebar s^{i+1}_{\ell}$.
\end{enumerate}
In addition, we connect the gadgets by the edges:
    \begin{enumerate}[label=(\alph*)]
        \item $(C^2_a, a_1^i)$ and $(C^1_a, \doublebar a_{1}^i)$,
        for each $i\in \set{0,\ldots,k-1}$;
            $(C^2_a, \singlebar a^0_{1})$ and $(C^2_a, \doublebar a^{k-1}_{1})$.
        \item $(C^2_b, b_1^i)$ and $(C^1_b, \doublebar b_{1}^i)$,
            for each $i\in \set{0,\ldots,k-1}$;
            $(C^2_b, \singlebar b^0_{1})$ and $(C^2_b, \doublebar b^{k-1}_{1})$.
        \item $(C^1_a, a_2^i)$ and $(C^2_a, \doublebar a_{2}^i)$,
            for each $i\in \set{0,\ldots,k-1}$;
            $(C^1_a, \singlebar a^0_{2})$ and $(C^1_a, \doublebar a^{k-1}_{2})$.
        \item $(C^1_b, b_2^i)$ and $(C^2_b, \doublebar b_{2}^i)$,
        for each $i\in \set{0,\ldots,k-1}$;
        $(C^1_b, \singlebar b^0_{2})$ and $(C^1_b, \doublebar b^{k-1}_{2})$.
    \end{enumerate}

Assume there is a proper $3$-coloring of $G$. Denote by $c_0,c_1$ and $c_2$ the colors of $C^0_a,C^1_a$ and $C^2_a$ respectively.
By construction, these are also the colors of $C^0_b,C^1_b$ and $C^2_b$, respectively.
For the nodes appearing in Section~\ref{sec:mvc},
coloring a node by $c_0$ is analogous to not including it in the vertex cover.

\begin{claim}
\label{claim: 3col c0 exists}
In each set $S\in \{A_1,A_2,B_1,B_2\}$, at least one node is colored by $c_0$.
\end{claim}

\begin{proof}
We start by proving the claim for $S=A_1$. Assume, towards a contradiction, that all nodes of $A_1$ are colored by $c_1$ and $c_2$.
All these nodes are connected to $C^2_a$, so they must all be colored by $c_1$.
Hence, all the nodes $\singlebar a^i_{1}$, $i\in \set{0,\ldots,k-1}$,
are colored by $c_0$ and $c_2$.
The nodes $\doublebar a^i_{1}$, $i\in \set{0,\ldots,k-1}$,
are connected to $C^1_a$,
so they are colored by $c_0$ and $c_2$ as well.

Hence, we have a path $(\singlebar a^0_{1}, \doublebar a^0_{1}, \singlebar a^1_{1}, \doublebar a^1_{1},\ldots \singlebar a^{k-1}_{1}, \doublebar a^{k-1}_{1})$
with an even number of nodes,
starting in $\singlebar a^0_{1}$ and ending in $\doublebar a^{k-1}_{1}$.
This path must be colored by alternating $c_0$ and $c_2$,
but both $\singlebar a^0_{1}$ and  $\doublebar a^{k-1}_{1}$ are connected to $C^2_a$,
so they cannot be colored by $c_2$, a contradiction.

A similar proof shows the claim for $S=B_1$.
For $S\in\set{A_2,B_2}$, we use a similar argument but change the roles of $c_1$ and $c_2$.
\end{proof}

\begin{claim}
\label{claim: 3col relating alices and bobs c0 nodes}
For each $i\in \set{0,\ldots, k-1}$,
the node $a_1^i$ is colored by $c_0$ iff $b_1^i$ is colored by $c_0$ and
the node $a_2^i$ is colored by $c_0$ iff $b_2^i$ is colored by $c_0$.
\end{claim}

\begin{proof}
Assume $a_1^i$ is colored by $c_0$,
so all of its adjacent nodes $\bin(a_1^i)$ can only be colored by $c_1$ or $c_2$.
As all of these nodes are connected to $C^a_1$, they must be colored by $c_2$.
Similarly, if a node $b_1^j$ in $B_1$ is colored by $c_0$,
then the nodes $\bin(b_1^j)$, which are also adjacent to $C^b_1$,
must be colored by $c_2$.

If $i\neq j$ then there must be a bit $i$ such that $i_h\neq j_h$, and there must be a pair of neighboring nodes $(f^h_{A_1},t^h_{B_1})$ or $(t^h_{A_1},f^h_{B_1})$ which are colored by $c_2$.
Thus, the only option is $i=j$.
By Claim~\ref{claim: 3col c0 exists}, there is a node in $B_1$ that is colored by $c_0$, and so it must be $b_1^i$.

An analogous argument shows that if $b_1^i$ is colored by $c_0$, then so does $a_1^i$.
For $a_2^i$ and $b_2^i$, similar arguments apply, where $c_1$ plays the role of $c_2$.
\end{proof}

~\\
\noindent\textbf{Adding edges corresponding to the strings $x$ and $y$:}
Given two bit strings $x,y\in\set{0,1}^{k^2}$, we augment the graph $G$ described above with additional edges, which defines $G_{x,y}$.

Assume $x$ and $y$ are indexed by pairs of the form $(i,j)\in \set{0,\ldots,k-1}^2$.
To construct $G_{x,y}$, add edges to $G$ by the following rules:
if $x_{i,j}=0$ then add the edge $(a_1^i,a_2^j)$, and if $y_{i,j}=0$ then add the edge $(b_1^i,b_2^j)$.
To prove that $\set{G_{x,y}}$ is a family of lower bound graphs, it remains to prove the next lemma.

\begin{lemma}
\label{lemma:threefamily}
The graph $G_{x,y}$ is $3$-colorable iff $\disj(x,y)=\false$.
\end{lemma}

\begin{proof}
For the first direction, assume $G_{x,y}$ is $3$-colorable, and denote the colors by $c_0, c_1$ and $c_2$, as before. By Claim~\ref{claim: 3col c0 exists},
there are nodes $a_1^i\in A_1$ and $a_2^j\in A_2$ that are both colored by $c_0$. Hence, the edge $(a_1^i,a_2^j)$ does not exist in $G_{x,y}$, implying $x_{i,j}=1$. By Claim~\ref{claim: 3col relating alices and bobs c0 nodes}, the nodes $b_1^i$ and $b_2^j$ are also colored $c_0$, so $y_{i,j}=1$ as well, giving that $\disj(x,y)=\false$, as needed.

For the other direction, assume $\disj(x,y)=\false$, i.e, there is an index $(i,j)\in \set{0,\ldots,k-1}^2$ such that $x_{i,j}=y_{i,j}=1$.
Consider the following coloring.
\begin{enumerate}
    \item Color $C_a^i$ and $C_b^i$ by $c_i$, for $i\in \set{0,1,2}$.
    \item Color the nodes $a_1^i, b_1^i, a_2^j$ and $b_2^j$ by $c_0$. Color the nodes $a_1^{i'}$ and $b_1^{i'}$, for $i'\neq i$, by $c_1$, and the nodes $a_2^{j'}$ and $b_1^{j'}$, for $j'\neq j$, by $c_2$.
    \item Color the nodes of $\bin(a_1^i)$ by $c_2$,
        and similarly color the nodes of $\bin(b_1^i)$ by $c_2$.
        Color the rest of the nodes in this gadget,
        i.e.\ $\bin(a_1^{k-i})$ and $\bin(b_1^{k-i})$, by $c_0$. Similarly, color $\bin(a_2^j)$ and $\bin(b_2^j)$ by $c_0$ and
        $\bin(a_2^{k-j})$ and $\bin(b_2^{k-j})$ by $c_1$.
    \item Finally, color the nodes of the forms $\singlebar s^i_{\ell}$ and $\doublebar s^i_{\ell}$ as follows.
    \begin{enumerate}
        \item Color $\singlebar a_{1}^i$ and $\singlebar b_{1}^i$ by $c_1$,
            all nodes $\singlebar a_{1}^{i'}$ and $\singlebar b_{1}^{i'}$ with $i'< i$ by $c_0$,
            and all nodes $\singlebar a_{1}^{i'}$ and $\singlebar b_{1}^{i'}$ with $i'> i$ by $c_2$.
        \item Similarly, color $\singlebar a_{2}^i$ and $\singlebar b_{2}^i$ by $c_2$,
            all nodes $\singlebar a_{2}^{i'}$ and $\singlebar b_{2}^{i'}$ with $i'< i$ by $c_0$,
            and all nodes $\singlebar a_{2}^{i'}$ and $\singlebar b_{2}^{i'}$ with $i'> i$ by $c_1$.
        \item Color all nodes $\doublebar a_{1}^{i'}$ and $\doublebar b_{1}^{i'}$ with $i'< i$ by $c_2$,
            and all nodes $\doublebar a_{1}^{i'}$ and $\doublebar b_{1}^{i'}$ with $i'\geq i$ by $c_0$.
        \item Similarly, color all nodes $\doublebar a_{2}^{i'}$ and $\doublebar b_{2}^{i'}$ with $i'< i$ by $c_1$,
            and all nodes $\doublebar a_{2}^{i'}$ and $\doublebar b_{2}^{i'}$ with $i'\geq i$ by $c_0$.
    \end{enumerate}
\end{enumerate}
It is not hard to verify that the suggested coloring is indeed a proper $3$-coloring of $G_{x,y}$, which completes the proof.
\end{proof}

Having constructed the family of lower bound graphs, we are now ready to prove Theorem~\ref{thm: coloring veriants lb}.
\begin{proofof}{Theorem~\ref{thm: coloring veriants lb}}
To complete the proof of Theorem~\ref{thm: coloring veriants lb},  we divide the nodes of $G$ (which are also the nodes of $G_{x,y}$) into two sets.
Let $V_A = A_1 \cup A_2 \cup F_{A_1} \cup T_{A_1} \cup F_{A_2} \cup T_{A_2} \cup \set{C_a^0,C_a^1,C_a^2} \cup \singlebar A_1 \cup \doublebar A_1 \cup \singlebar A_2 \cup \doublebar A_2$, and $V_B=V\setminus V_A$.
Note that $n\in \Theta(k)$.

The edges in the cut $E(V_A, V_B)$ are the $6$ edges connecting $\set{C_a^0,C_a^1,C_a^2}$ and $\set{C_b^0,C_b^1,C_b^2}$, and $2$ edges for every $4$-cycle of the nodes of $F_{A_1} \cup T_{A_1} \cup F_{B_1} \cup T_{B_1}$ and $F_{A_2} \cup T_{A_2} \cup F_{B_2} \cup T_{B_2}$,
for a total of $\Theta(\log n)$ edges.
Since Lemma~\ref{lemma:threefamily} shows that $\{G_{x,y}\}$ is a family of lower bound graphs with respect to $\disj_K$, $K=k^2\in\Theta(n^2)$
and the predicate $\chi\leq3$, we can apply Theorem~\ref{thm: general lb framework} on the above partition to deduce that any algorithm in the CONGEST model for deciding whether a given graph is $3$-colorable requires at least $\Omega(n^2/\log^2n)$ rounds.

Any algorithm that computes $\chi$ of the input graph,
or produces a $\chi$-coloring of it,
may be used to deciding whether $\chi\leq3$,
in at most $O(D)$ additional rounds.
Thus, the lower bound applies to these problems as well.

Our construction and proof naturally extend to handle $c$-coloring,
for any $c\geq 3$.
To this end, we add to $G$ (and to $G_{x,y}$)
new nodes denoted $C_a^i$, $i\in\set{3,\ldots,c-1}$,
and connect them to all of $V_A$,
and new nodes denoted $C_b^i$, $i\in\set{3,\ldots,c-1}$,
and connect them to all of $V_B$ and also to $C_a^0, C_a^1$ and $C_a^2$.
The nodes $C_a^i$ are added to $V_a$, and the rest are added to $V_b$,
which increases the cut size by $\Theta(c)$ edges.

Assume the extended graph is colorable by $c$ colors,
and denote by $c_i$ the color of the node $C_a^i$
(these nodes are connected by a clique, so their colors are distinct).
The nodes $C_b^i$, $i\in\set{2,\ldots,c-1}$ form a clique,
and they are all connected to the nodes $C_a^0,C_a^1$ and $C_a^2$,
so they are colored by the colors $\set{c_3,\ldots,c_{c-1}}$,
in some order.
All the original nodes of $V_A$ are connected to
$C_a^i$, $i\in\set{3,\ldots,c-1}$,
and all the original nodes of $V_B$ are connected to
$C_b^i$, $i\in\set{3,\ldots,c-1}$,
so the original graph must be colored by $3$ colors,
which we know is possible iff $\disj(x,y)=\false$.

We added $2c-6$ nodes to the graph,
so the inputs strings are of length $K=n-2c+6$.
Thus, the new graphs constitute a family of lower bound graphs
with respect to $\eq_{K}$ and the predicate $\chi\leq c$,
the communication complexity of $\eq_{K}$
is in $\Omega(K^2)=\Omega((n-c)^2)$,
the cut size is $\Theta(c+\log n)$,
and Theorem~\ref{thm: general lb framework} completes the proof.
\end{proofof}

~\\
\noindent\textbf{A lower bound for $(4/3-\epsilon)$-approximation:} Finally, we extend our construction to give a lower bound
for approximate coloring.
That is,
we show a similar lower bound for computing a $(4/3-\eps)$-approximation to $\chi$
and for finding a coloring in $(4/3-\eps)\chi$ colors.

Observe that since $\chi$ is integral, any $(4/3-\epsilon)$-approximation algorithm must return the exact solution in case $\chi=3$. Thus, in order to rule out the possibility for an algorithm which is allowed to return a $(4/3-\eps)$-approximation which is not the exact solution, we need a more general construction.
For any integer $c$,
we show a lower bound for distinguishing between the case
$\chi\leq 3c$ and $\chi\geq 4c$.

\begin{claim}\label{claim: approx-coloring}
Given an integer $c$, any distributed algorithm in the CONGEST model that distinguishes a graph $G$ with $\chi(G)\leq 3c$ from a graph with $\chi(G)\geq 4c$ requires $\Omega(n^2/(c^3\log^2n))$ rounds.
\end{claim}

To prove Claim~\ref{claim: approx-coloring} we show a family of lower bound graphs with respect to the $\disj_{K}$ function, where $K\in \Theta(n^2/c^2)$,
and the predicate $\chi\leq3c$ ($\true$)
or $\chi\geq 4c$ ($\false$).
The predicate is not defined for other values of $\chi$.

We create a graph $G^{c}_{x,y}$, composed of $c$ copies of $G_{x,y}$.
The $i$-th copy is denoted $G_{x,y}(i)$,
and its nodes are partitioned into $V_A(i)$ and $V_B(i)$.
We connect all the nodes of $V_A(i)$ to all nodes of $V_A(j)$,
for each $i\neq j$.
Similarly, we connect all the nodes of $V_B(i)$ to all the nodes of $V_B(j)$.
This construction guarantees that each copy is colored by different colors, and hence
if $\disj(x,y)=\false$ then $\chi(G^{c}_{x,y})=3c$
and otherwise $\chi(G^{c}_{x,y})\geq3c$.
Therefore, $G^{c}_{x,y}$ is a family of lower bound graphs.

\begin{proofof}{Claim~\ref{claim: approx-coloring}}
Note that $n\in \Theta(kc)$. Thus, $K=|x|=|y|=n^2/c^2$.
Furthermore, observe that for each $G_{x,y}(i)$, there are $O(\log n)$ edges in the cut,
so in total $G^{c}_{x,y}$ contains $O(c\log n)$ edges in the cut.
Since we showed that $G^{c}_{x,y}$ is a family of lower bound graphs,
we can apply Theorem~\ref{thm: general lb framework} to deduce that because of the lower bound for Set Disjointness,
any algorithm in the CONGEST model for distinguishing between $\chi\leq 3c$ and $\chi\geq 4c$ requires at least $\Omega(n^2/c^3\log^2(n))$ rounds.
\end{proofof}

For any $\epsilon>0$ and any $c$ it holds that $(4/3-\epsilon)3c<4c$.
Thus, we can choose $c$ to be an arbitrary constant to achieve the following theorem.

\begin{theorem}
\label{cor: 3col approx lb}
For any constant $\eps>0$,
any distributed algorithm in the CONGEST model
that computes a $(4/3-\eps)$-approximation to $\chi$
requires $\Omega(n^2/\log^2n)$ rounds.
\end{theorem}
}

\ifabs{
}
{\subsection{Graph Coloring}
\label{sec:coloring}
\ColoringSection
}

\newcommand{\Psection}{
In this section we support our claim that what makes problems hard for the CONGEST model is not necessarily them being NP-hard problems. First, we address a class of subgraph detection problems, which requires detecting cycles of length $8$ and a given weight, and show a near-quadratic lower bound on the number of rounds required for solving it, although its sequential complexity is polynomial. Then, we define a problem which we call the \emph{Identical Subgraphs Detection} problem, in which the goal is to decide whether two given subgraphs are identical. While this last problem is rather artificial, it allows us to obtain a strictly quadratic lower bound for the CONGEST model, with a problem that requires only a single-bit output.

\subsection{Weighted Cycle Detection}
\label{sec:cycle}

In this section we show a lower bound on the number of rounds needed in order to decide the graph contains a cycle of length $8$ and weight $W$, such that $W$ is a $\mbox{polylog}(n)$-bit value given as an input. Note that this problem can be solved easily in polynomial time in the sequential setting by simply checking all of the at most $n\choose 8$ cycles of length $8$.

\begin{theorem}\label{thm: wCycles}
	Any distributed algorithm in the CONGEST model that decides if a graph with edge weights $w:E\rightarrow[0,\mbox{poly}(n)]$ contains a cycle of length $8$ and weight $W$ requires $\Omega(n^2/\log^2n)$ rounds.
\end{theorem}
Similarly to the previous sections, to prove Theorem~\ref{thm: wCycles} we describe a family of lower bound graphs with respect to the Set Disjointness function and the predicate $P$ that says that the graph contains a cycle of length $8$ and weight $W$.

~\\
\noindent\textbf{The fixed graph construction:}
The fixed graph construction $G=(V,E)$ is defined as follows. The set of nodes contains four sets $A_1,A_2,B_1$ and $B_2$, each of size $k$. To simplify our proofs in this section, we assume that $k\geq 3$. For each set $S\in \{A_1,A_2,B_1,B_2\}$ there is a node $c_S$, which is connected to each of the nodes in $S$ by an edge of weight $0$. In addition there is an edge between $c_{A_1}$ and $c_{B_1}$ of weight 0 and an edge between $c_{A_2}$ and $c_{B_2}$ of weight 0 (see Figure~\ref{fig: cycle}).
\begin{figure}[t]
	\begin{center}
		\includegraphics[scale=0.5]{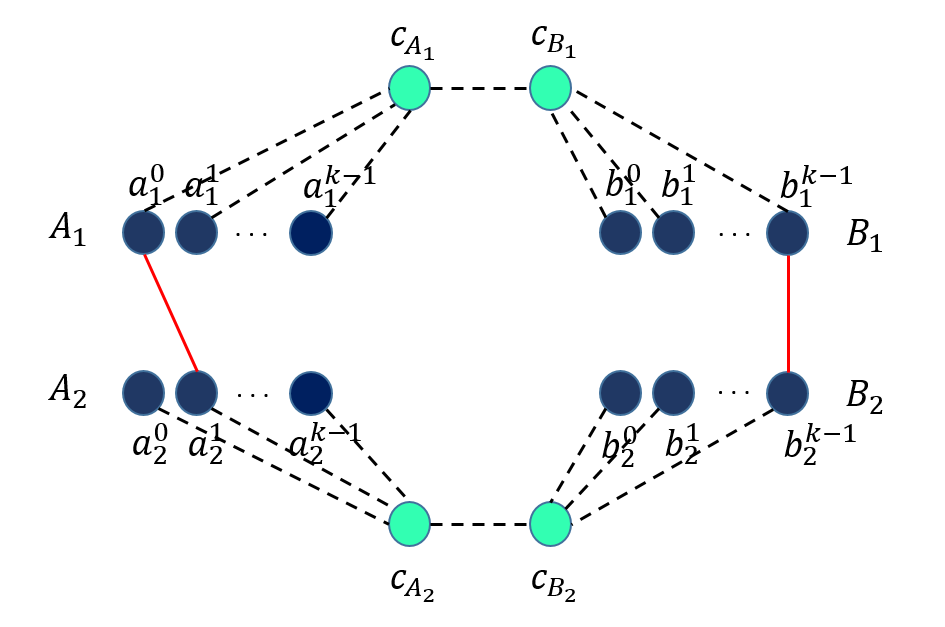}
	\end{center}
	\caption{{\small The family of lower bound graphs for detecting weighted cycles, with many edges omitted for clarity. Here, $x_{0,1}=1$ and $y_{k-1,k-1}=1$. Thus, $a^0_1$ is connected to $a^1_2$ by an edge of weight $k^3+k\cdot0+1=k^3+1$, and $b^{k-1}_1$ is connected to $b^{k-1}_2$ by an edge of weight $k^3-(k(k-1)+k-1)=k^3-k^2+1$. All the dashed edges are of weight 0.}}
	\label{fig: cycle}
\end{figure}

~\\
\noindent\textbf{Adding edges corresponding to the strings $x$ and $y$:}
Given two binary strings $x,y\in\set{0,1}^{k^2}$, we augment the fixed graph $G$ defined in the previous section with additional edges, which defines $G_{x,y}$. Recall that we assume that $k\geq 3$. Let $x$ and $y$ be indexed by pairs of the form $(i,j)\in \set{0,\ldots,k-1}^2$. For each $(i,j)\in \set{0,\ldots,k-1}^2$, we add to $G_{x,y}$ the following edges. If $x_{i,j}=1$, then we add an edge of weight $k^3+ki +j$ between the nodes $a_1^i$ and $a_2^j$. If $y_{i,j}=1$, then we add an edge of weight $k^3-(ki+j)$ between the nodes $b_1^i$ and $b_2^j$. We denote by $InputEdges$ the set of edges $\{(u,v)\mid u\in A_1 \wedge v\in A_2\}\cup \{(u,v)\mid u\in B_1 \wedge v\in B_2\}$, and we denote by $w(u,v)$ the weight of the edge $(u,v)$.

Observe that the graph does not contain edges of negative weight. Furthermore, the weight of any edge in $InputEdges$ does not exceed $k^3+k^2-1$, which is the weight of the edge $(a^{k-1}_1,a^{k-1}_2)$, in case $x_{k-1,k-1}=1$. Similarly, the weight of an edge in $InputEdges$ is not less than $k^3-k^2+1$, which is the weight of the edge $(b^{k-1}_1,b^{k-1}_2)$, in case $y_{k-1,k-1}=1$. Using these two simple observations, we deduce the following claim.

\begin{claim}\label{2-edges-cycle}
	For any cycle of weight $2k^3$, the number of edges it contains that are in $InputEdges$ is exactly two.
\end{claim}
\begin{proof}
	Let $C$ be a cycle of weight $2k^3$, and assume for the sake of contradiction that $C$ does not contain exactly two edges from $InputEdges$. In case $C$ contains exactly one edge from $InputEdges$, then the weight of $C$ is at most $k^3+k^2-1< 2k^3$, because all the other edges of $C$ are of weight 0. Otherwise, in case $C$ contains three or more edges from $InputEdges$, it holds that the weight of $C$ is at least $3k^3-3k^2+3>2k^3$, because all the other edges on $C$ are of non-negative weight.
\end{proof}

To prove that $\set{G_{x,y}}$ satisfies the definition of a family of lower bound graphs,
we prove the following lemma.

\begin{lemma}\label{lemma: maincycles}
	The graph $G_{x,y}$ contains a cycle of length $8$ and weight $W=2k^3$ if and only if $\disj(x,y)=\false$.
\end{lemma}
\begin{proof}
	For the first direction, assume that $\disj(x,y)=\false$ and let $0\leq i,j\leq k-1$ be such that $x_{i,j}=1$ and $y_{i,j}=1$. Consider the cycle $(a^i_1,c_{A_1},c_{B_1},b^i_1,b^j_2,c_{B_2},c_{A_2},a^j_2)$. It is easy to verify that this is a cycle of length 8 and weight $w(a^j_1,a^i_2)+w(b^i_1,b^j_2)=k^3+ki+j+k^3-ki-j=2k^3$, as needed.
	
	For the other direction, assume that the graph contains a cycle $C$ of length $8$ and weight $2k^3$. By Claim~\ref{2-edges-cycle}, $C$ contains exactly two edges from $InputEdges$. Denote these two edges by $(u_1,v_1)$ and $(u_2,v_2)$. Since all the other edges in $C$ are of weight 0, the weight of $C$ is $w(u_1,v_1)+w(u_2,v_2)$. The rest of the proof is by case analysis, as follows. First, it is not possible that $(u_1,v_1),(u_2,v_2)\in \{(u,v)\mid u\in A_1 \wedge v\in A_2\}$, since in this case $w(u_1,v_1)+w(u_2,v_2) \geq w(a^0_1,a^0_2)+w(a^0_1,a^1_2)=2k^3+1$. Similarly, it is not possible that $(u_1,v_1),(u_2,v_2)\in \{(u,v)\mid u\in B_1 \wedge v\in B_2\}$, since in this case $w(u_1,v_1)+w(u_2,v_2)\leq w(b^{0}_1,b^{0}_2)+w(b^{0}_1,b^{1}_2)=2k^3-1$. Finally, suppose without loss of generality that $(u_1,v_1)\in \{(u,v)\mid u\in A_1 \wedge v\in A_2\}$ and $(u_2,v_2)\in \{(u,v)\mid u\in B_1 \wedge v\in B_2\}$. Denote $u_1=a^i_1, u_2=a^j_1, v_1=b^{i'}_1$ and $v_2=b^{j'}_2$. It holds that $w(a^i_1,a^j_2)+w(b^{i'}_1,b^{j'}_2)=2k^3$ if and only if $i=i'$ and $j=j'$, which implies that $x_{i,j}=1$ and $y_{i,j}=1$ and $\disj(x,y)=\false$.
\end{proof}

Having constructed the family of lower bound graphs, we are now ready to prove Theorem~\ref{thm: wCycles}.
\begin{proofof}{Theorem~\ref{thm: wCycles}}
	To complete the proof of Theorem~\ref{thm: wCycles}, we divide the nodes of $G$ (which are also the nodes of $G_{x,y}$) into two sets. Let $V_A = A_1 \cup A_2 \cup \{c_{A_1},c_{A_2}\}$ and $V_B=V\setminus V_A$. Note that $n\in \Theta(k)$. Thus, $K=|x|=|y|=\Theta(n^2)$. Furthermore, note that the only edges in the cut $E(V_A, V_B)$ are the edges $(c_{A_1},c_{B_1})$ and $(c_{A_2},c_{B_2})$.
	Since Lemma~\ref{lemma: maincycles} shows that $\{G_{x,y}\}$ is a family of lower bound graphs, we apply Theorem~\ref{thm: general lb framework} on the above partition to deduce that because of the lower bound for Set Disjointness, any algorithm in the CONGEST model for deciding whether a given graph contains a cycle of length $8$ and weight $W=2k^3$ requires at least $\Omega(K/\log n)=\Omega(n^2/\log n)$ rounds.
\end{proofof}

\subsection{Identical Subgraphs Detection}
In this section we show the strongest possible, quadratic lower bound,
for a problem which can be solved in linear time in the sequential setting.

Consider the following sequential specification of a graph problem.

\begin{definition}(The Identical Subgraphs Detection Problem)
\label{def:identicalproblem}\newline
Given a weighted input graph $G=(V,E,w)$, with an edge-weight function $w:E \to \set{0,\ldots,W-1}$, $W\in \poly n$, such that the set of nodes $V$ is partitioned into two enumerated sets of the same size, $V_A=\{a_0,...,a_{k-1}\}$ and $V_B=\{b_0,...,b_{k-1}\}$, the \emph{Identical Subgraphs Detection} problem is to determine whether the subgraph induced by the set $V_A$ is identical to the subgraph induced by the set $V_B$, in the sense that for each $0\leq i,j\leq k-1$ it holds that $(a_i,a_j)\in E$ if and only if $(b_i,b_j)\in E$ and $w(a_i,a_j)=w(b_i,b_j)$ if these edges exist.
\end{definition}

The identical subgraphs detection problem can be solved easily in linear time in the sequential setting by a single pass over the set of edges. However, as we prove next, it requires a quadratic number of rounds in the CONGEST model, in any deterministic solution (note that this restriction did not apply in the previous sections).
For clarity, we emphasize that in the distributed setting, the input to each node in the identical subgraphs detection problem is its enumeration $a_i$ or $b_i$, as well as the enumerations of its neighbors and the weights of the respective edges.
The outputs of all nodes should be $\true$ if the subgraphs are identical,
and $\false$ otherwise.

\begin{theorem}\label{thm: Identical Subgraphs}
	Any distributed deterministic algorithm in the CONGEST model for solving the identical subgraphs detection problem requires $\Omega(n^2)$ rounds.
\end{theorem}

To prove Theorem~\ref{thm: Identical Subgraphs} we describe a family of lower bound graphs.

~\\
\noindent\textbf{The fixed graph construction:}
The fixed graph $G$ is composed of two $k$-node cliques
on the node sets $V_A=\{a_0,...,a_{k-1}\}$ and $V_B=\{b_0,...,b_{k-1}\}$,
and one extra edge $(a_0,b_0)$.

~\\
\noindent\textbf{Adding edge weights corresponding to the strings $x$ and $y$:}
Given two binary strings $x$ and $y$, each of size $K={k \choose 2}\log n $, we augment the graph $G$ with additional edge weights, which define $G_{x,y}$. For simplicity, assume that $x$ and $y$ are vectors of $\log n $-bit numbers each having ${k\choose 2}$ entries enumerated as $x_{i,j}$ and $y_{i,j}$, with $i<j$, $i,j\in \set{0,\ldots,k-1}$. For each such $i$ and $j$ we set the weights of $w(a_i,a_j)=x_{i,j}$ and $w(b_i,b_j)=y_{i,j}$,
and we set $w(a_0,b_0)=0$.
 Note that $\set{G_{x,y}}$ is a family of lower bound graphs with respect to $\eq_K$ and the predicate $P$ that says that the subgraphs are identical in the aforemention sense.

\begin{proofof}{Theorem~\ref{thm: Identical Subgraphs}}
Note that $n\in \Theta(k)$, and thus $K=|x|=|y|=\Theta(n^2\log n)$. Furthermore, the only edge in the cut $E(V_A, V_B)$ is the edge $(a_0,b_0)$. Since we showed that $\set{G_{x,y}}$ is a family of lower bound graphs, we can apply Theorem~\ref{thm: general lb framework} on the above partition to deduce that because of the lower bound for $\eq_K$, any deterministic algorithm in the CONGEST model for solving the identical subgraphs detection problem requires at least $\Omega(K/\log n)=\Omega(n^2)$ rounds.
\end{proofof}

We remark that in a distributed algorithm for the identical subgraphs detection problem running on our family of lower bound graphs,
information about essentially all the edges and weights in the subgraphs induced on $V_A$ and $V_B$ needs to be sent across the edge $(a_0,b_0)$. This might raise the suspicion that this problem is reducible to learning the entire graph, making the lower bound trivial.
To argue that this is far from being the case,
we present a randomized algorithm that solves the identical subgraphs detection problem in $O(D)$ rounds and succeeds w.h.p.
This has the additional benefit of providing the strongest possible separation between deterministic and randomized complexities for global problems in the CONGEST model, as the former is $\Omega(n^2)$ and the latter is at most $O(D)$.

Our starting point is the following randomized algorithm for the $\eq_K$ problem, presented in~\cite[Exersise 3.6]{KushilevitzN:book96}. Alice chooses a prime number $p$ among the first $K^2$ primes uniformly at random. She treats her input string $x$ as a binary representation of an integer $\bar x=\sum_{\ell=0}^{K-1} 2^\ell x_\ell$, and sends $p$ and $\bar x\pmod p$ to Bob. Bob similarly computes $\bar y$, compares $\bar x\bmod p$ with $\bar y\bmod p$, and returns $\true$ if they are equal and false otherwise. The error probability of this protocol is at most $1/K$.

We present a simple adaptation of this algorithm for the identical subgraph detection problem. Consider the following encoding of a weighted induced subgraph on $V_A$: for each pair $i,j$ of indices, we have $\lceil\log W \rceil+1$ bits, indicating the existence of the edge and its weight (recall that $W\in\poly n$ is an upper bound on the edge weights). This weighted induced subgraph is thus represented by a $K\in O(n^2\log n)$ bit-string, denoted $x = x_0,\ldots, x_{K-1}$, and each pair $(i,j)$ has a set $S_{i,j}$ of indices representing the edge $(a_i,a_j)$.
Note that the bits $\set{x_\ell\mid \ell\in s_{i,j}}$ are known to both $a_i$ and $a_j$, and in the algorithm we use the node with smaller index in order to encode these bits. Similarly, a $K\in O(n^2\log n)$ bit-string, denoted $y = y_0,\ldots, y_{K-1}$ encodes a weighted induced subgraph on $V_B$.

\textbf{The Algorithm.}
As standard, assume the input graph is connected. The nodes are enumerated as in Definition~\ref{def:identicalproblem}. The algorithm starts with some node, say, $a_0$, constructing a BFS tree, which completes in $O(D)$ rounds. Then, $a_0$ chooses a prime number $p$ among the first $K^2$ primes uniformly at random and sends $p$ to all the nodes over the tree, which takes $O(D)$ rounds.

Each node $a_i$ computes the sum $\sum_{j>i}\sum_{\ell\in S_{i,j}} x_\ell 2^\ell \mod p$, and the nodes then aggregate these local sums modulo $p$ up the tree, until $a_0$ computes the sum
$\bar x\mod p = \sum_{j\neq i}\sum_{\ell\in S_{i,j}} x_\ell 2^\ell \mod p$. A similar procedure is then invoked w.r.t $\bar{y}$. Finally, $a_0$ compares $\bar x\mod p$ and $\bar y\mod p$, and downcasts over the BFS tree its output, which is $\true$ if these values are equal and is $\false$ otherwise.

If the subgraphs are identical, $a_0$ always returns $\true$, while otherwise their encoding differs in at least one bit, and as in the case of $\eq_K$, $a_0$ returns $\true$ falsely with probability at most $1/K\in O(1/n^2)$.

\begin{theorem}
\label{thm: Identical Subgraphs randomized alg}
There is a randomized algorithm in the CONGEST model that 
solves the identical subgraphs detection problem
on any connected graph
in $O(D)$ rounds.
\end{theorem}
}

\ifabs{}{
\section{Quadratic and Near-Quadratic Lower Bounds for Problems in P}
\label{sec:P}
\Psection
}

\section{Weighted APSP}
\label{sec:APSP}

In this section we use the following, natural extension of Definition~\ref{def:family}, in order to address more general 2-party functions, as well as distributed problems that are not decision problems.

For a function $f:\set{0,1}^{K_1}\times\set{0,1}^{K_2}\to \set{0,1}^{L_1}\times\set{0,1}^{L_2}$, we define a family of lower bound graphs in a similar way as Definition~\ref{def:family}, except that we replace item (4) in the definition with a generalized requirement that says that for $G_{x,y}$, the values of the of nodes in $V_A$ uniquely determine the left-hand side of $f(x,y)$, and the values of the of nodes in $V_B$ determine the right-hand side of $f(x,y)$. Next, we argue that theorem similar to Theorem~\ref{thm: general lb framework} holds for this case.

\begin{theorem}
\label{thm: general lb framework APSP}
Fix a function $f:\set{0,1}^{K_1}\times\set{0,1}^{K_2}\to \set{0,1}^{L_1}\times\set{0,1}^{L_2}$ and a graph problem $P$. If there is a family $\set{G_{x,y}}$ of lower bound graphs with $C = E(V_A, V_B)$ then any deterministic algorithm for solving $P$ in the CONGEST model requires $\Omega (CC(f)/\size{C}\log n)$ rounds, and any randomized algorithm for deciding $P$ in the CONGEST model requires $\Omega (CC^R(f)/\size{C}\log n)$ rounds.
\end{theorem}

The proof is similar to that of Theorem~\ref{thm: general lb framework}. Notice that the only difference between the theorems, apart from the sizes of the inputs and outputs of $f$, are with respect to item (4) in the definition of a family of lower bound graphs. However, the essence of this condition remains the same and is all that is required by the proof: the values that a solution to $P$ assigns to nodes in $V_A$ determines the output of Alice for $f(x,y)$, and the values that a solution to $P$ assigns to nodes in $V_B$ determines the output of Bob for $f(x,y)$.

\subsection{A Linear Lower Bound for Weighted APSP}\label{sec:linearlb}

Nanongkai~\cite{Nanongkai14} showed that any algorithm in the CONGEST model for computing a $\poly(n)$-approximation for weighted all pairs shortest paths (APSP) requires at least $\Omega(n/\log n)$ rounds. In this section we show that a slight modification to this construction yields an $\Omega(n)$ lower bound for computing exact weighted APSP. As explained in the introduction, this gives a separation between the complexities of the weighted and unweighted versions of APSP. At a high level, while we use the same simple topology for our lower bound as in~\cite{Nanongkai14}, the reason that we are able to shave off the extra logarithmic factor is because our construction uses $O(\log{n})$ bits for encoding the weight of each edge out of many optional weights, while in~\cite{Nanongkai14} only a single bit is used per edge for encoding one of only two options for its weight.

\begin{theorem}\label{Thm:wAPSP}
Any distributed algorithm in the CONGEST model for computing exact weighted all pairs shortest paths requires at least $\Omega(n)$ rounds.
\end{theorem}

The reduction is from the following, perhaps simplest, $2$-party communication problem. Alice has an input string $x$ of size $K$ and Bob needs to learn the string of Alice. Any algorithm (possibly randomized) for solving this problem requires at least $\Omega(K)$ bits of communication, by a trivial information theoretic argument.

Notice that the problem of having Bob learn Alice's input is not a binary function as addressed in Section~\ref{sec:preliminaries}. Similarly, computing weighted APSP is not a decision problem, but rather a problem whose solution assigns a value to each node (which is its vector of distances from all other nodes). We therefore use the extended Theorem~\ref{thm: general lb framework APSP} above.

~\\
\noindent\textbf{The fixed graph construction:}
The fixed graph construction $G = (V,E)$ is defined as follows.
It contains a set of $n-2$ nodes, denoted $A=\{a_0,...,a_{n-3}\}$,
which are all connected to an additional node $a$.
The node $a$ is connected to the last node $b$, by an edge of weight 0.

~\\
\noindent\textbf{Adding edge weights corresponding to the string $x$:} Given the binary string $x$ of size $K=(n-2)\log n$ we augment the graph $G$ with edge weights, which defines $G_{x}$, by having each non-overlapping batch of $\log{n}$ bits encode a weight of an edge from $A$ to $a$. It is straightforward to see that $G_{x}$ is a family of lower bound graphs for a function $f$ where $K_2=L_1=0$, since the weights of the edges determine the right-hand side of the output (while the left-hand side is empty).

\begin{proofof}{Theorem~\ref{Thm:wAPSP}}
To prove Theorem~\ref{Thm:wAPSP}, we let $V_A = A\cup\{a\}$ and $V_B=\set{b}$.
Note that $K=|x|=\Theta(n\log n)$. Furthermore, note that the only edge in the cut $E(V_A, V_B)$ is the edge $(a,b)$.
Since we showed that $\{G_{x}\}$ is a family of lower bound graphs, we apply Theorem~\ref{thm: general lb framework APSP} on the above partition to deduce that because $K$ bits are required to be communicated in order for Bob to know Alice's $K$-bit input, any algorithm in the CONGEST model for computing weighted APSP requires at least $\Omega(K /\log n)=\Omega(n)$ rounds.
\end{proofof}

\subsection{The Alice-Bob Framework Cannot Give a Super-Linear Lower Bound for Weighted APSP}
\label{sec:APSPalicebob}
In this section we argue that a reduction from any 2-party function with a constant partition of the graph into Alice and Bob's sides is provable incapable of providing a super-linear lower bound for computing weighted all pairs shortest paths in the CONGEST model.
A more detailed inspection of our analysis shows a stronger claim:
our claim also holds for algorithms for the CONGEST-BROADCAST model, where in each round each node must send the same $(\log{n})$-bit message to all of its neighbors. The following theorem states our claim.

\begin{theorem}
\label{thm:noAliceBob}
Let $f:\set{0,1}^{K_1}\times\set{0,1}^{K_2}\to \set{0,1}^{L_1}\times\set{0,1}^{L_2}$ be a function and let $G_{x,y}$ be a family of lower bound graphs w.r.t.\ $f$ and the weighted APSP problem. When applying Theorem~\ref{thm: general lb framework APSP} to $f$ and $G_{x,y}$, the lower bound obtained for the number of rounds for computing weighted APSP is at most linear in $n$.
\end{theorem}

Roughly speaking, we show that given an input graph $G=(V,E)$ and a partition of the set of vertices into two sets $V=V_A\cup V_B$, such that the graph induced by the nodes in $V_A$ is simulated by Alice and the graph induced by nodes in $V_B$ is simulated by Bob, Alice and Bob can compute weighted all pairs shortest paths by communicating $O(n\log n)$ bits of information for each node touching the cut $C=(V_A,V_B)$ induced by the partition. This means that for any 2-party function $f$ and any family of lower bound graphs w.r.t. $f$ and weighted APSP according to the extended definition of Section~\ref{sec:linearlb}, since Alice and Bob can compute weighted APSP which determines their output for $f$ by exchanging only $O(|V(C)|n\log n)$ bits, where $V(C)$ is the set of nodes touching $C$, the value $CC(f)$ is at most $O(|V(C)|n\log n)$. But then the lower bound obtained by Theorem~\ref{thm: general lb framework APSP} cannot be better than $\Omega(n)$, and hence no super-linear lower can be deduced by this framework as is.

Formally, given a graph $G=(V=V_A\dot\cup V_B,E)$
we denote $C=E(V_A,V_B)$. Let $V(C)$ denote the nodes touching the cut $C$, with $C_A=V(C)\cap V_A$ and $C_B=V(C)\cap V_B$.
Let $G_A=(V_A,E_A)$ be the subgraph induced by the nodes in $V_A$ and let $G_B=(V_B,E_B)$ be the subgraph induced by the nodes in $V_B$.
For a graph $H$, we denote the weighted distance between two nodes $u,v$ by $\wdist_{H}(u,v)$.

\begin{lemma}
\label{thm:AliceBobCompute}
Let $G=(V=V_A\dot\cup V_B,E,w)$ be a graph with an edge-weight function $w:E\to \set{1,\ldots, W}$, such that $W\in \poly n$. Suppose that $G_A$,  $C_B$, $C$ and the values of $w$ on $E_A$ and $C$ are given as input to Alice, and that $G_B$, $C_A$, $C$ and the values of $w$ on $E_B$ and $C$ are given as input to Bob.

Then, Alice can compute the distances in $G$ from all nodes in $V_A$ to all nodes in $V$ and Bob can compute the distances from all nodes in $V_B$ to all the nodes in $V$, using $O(\size{V(C)}n\log n)$ bits of communication.
\end{lemma}

\ifabs{
\begin{proof}
We describe a protocol for the required computation, as follows. For each node $u\in C_B$, Bob sends to Alice the weighted distances in $G_B$ from $u$ to all nodes in $V_B$, that is, Bob sends $\{\wdist_{G_B}(u,v) \mid u\in C_B, v\in V_B\}$
(or $\infty$ for pairs of nodes not connected in $G_B$).
Alice constructs a virtual graph $G_A'=(V_A',E_A',w_A')$ with the nodes $V_A' = V_A\cup C_B$ and edges $E_A'=E_A\cup C \cup (C_B\times C_B)$. The edge-weight function $w_A'$ is defined by $w_A'(e)=w(e)$ for each $e\in E_A\cup C$, and $w_A'(u,v)$ for $u,v\in C_B$ is defined to be the weighted distance between $u$ and $v$ in $G_B$, as received from Bob. Alice then computes the set of all weighted distances in $G_A'$, $\{\wdist_{G_A'}(u,v) \mid u, v\in V_A'\}$.

Alice assigns her output for the weighted distances in $G$ as follows. For two nodes $u,v\in V_A\cup C_B$, Alice outputs their weighted distance in $G_A'$, $\wdist_{G_A'}(u,v)$. For a node $u\in V_A'$ and a node $v\in V_B\setminus C_B$, Alice outputs $\min\{\wdist_{G_A'}(u,x)+\wdist_{G_B}(x,v)\mid x\in C_B\}$, where $\wdist_{G_A'}$ is the distance in $G_A'$ as computed by Alice, and $\wdist_{G_B}$ is the distance in $G_B$ that was sent by Bob.

For Bob to compute his required weighted distances, for each node $u\in C_A$, similar information is sent by Alice to Bob, that is, Alice sends to Bob the weighted distances in $G_A$ from $u$ to all nodes in $V_A$. Bob constructs the analogous graph $G_B'$ and outputs his required distance. The next paragraph formalizes this for completeness, but may be skipped by a convinced reader.

Formally, Alice sends $\{\wdist_{G_A}(u,v) \mid u\in C_A, v\in V_A\}$. Bob constructs $G_B'=(V_B',E_B',w_B')$ with $V_B' = V_B\cup C_A$ and edges $E_B'=E_B\cup C \cup (C_A\times C_A)$. The edge-weight function $w_B'$ is defined by $w_B'(e)=w(e)$ for each $e\in E_B\cup C$, and $w_B'(u,v)$ for $u,v\in C_A$ is defined to be the weighted distance between $u$ and $v$ in $G_A$, as received from Alice (or $\infty$ if they are not connected in $G_A$). Bob then computes the set of all weighted distances in $G_B'$, $\{\wdist_{G_B'}(u,v) \mid u, v\in V_B'\}$. Bob assigns his output for the weighted distances in $G$ as follows. For two nodes $u,v\in V_B\cup C_A$, Bob outputs their weighted distance in $G_B'$, $\wdist_{G_B'}(u,v)$. For a node $u\in V_B'$ and a node $v\in V_A\setminus C_A$, Bob outputs $\min\{\wdist_{G_B'}(u,x)+\wdist_{G_A}(x,v)\mid x\in C_A\}$, where $\wdist_{G_B'}$ is the distance in $G_B'$ as computed by Bob, and $\wdist_{G_A}$ is the distance in $G_A$ that was sent by Alice.

The correctness and complexity of the algorithm are addressed in Appendix~\ref{append:AliceBobCompute}.
\end{proof}
}{
\begin{proof}
	We describe a protocol for the required computation, as follows. For each node $u\in C_B$, Bob sends to Alice the weighted distances in $G_B$ from $u$ to all nodes in $V_B$, that is, Bob sends $\{\wdist_{G_B}(u,v) \mid u\in C_B, v\in V_B\}$
	(or $\infty$ for pairs of nodes not connected in $G_B$).
	Alice constructs a virtual graph $G_A'=(V_A',E_A',w_A')$ with the nodes $V_A' = V_A\cup C_B$ and edges $E_A'=E_A\cup C \cup (C_B\times C_B)$. The edge-weight function $w_A'$ is defined by $w_A'(e)=w(e)$ for each $e\in E_A\cup C$, and $w_A'(u,v)$ for $u,v\in C_B$ is defined to be the weighted distance between $u$ and $v$ in $G_B$, as received from Bob. Alice then computes the set of all weighted distances in $G_A'$, $\{\wdist_{G_A'}(u,v) \mid u, v\in V_A'\}$.
	
	Alice assigns her output for the weighted distances in $G$ as follows. For two nodes $u,v\in V_A\cup C_B$, Alice outputs their weighted distance in $G_A'$, $\wdist_{G_A'}(u,v)$. For a node $u\in V_A'$ and a node $v\in V_B\setminus C_B$, Alice outputs $\min\{\wdist_{G_A'}(u,x)+\wdist_{G_B}(x,v)\mid x\in C_B\}$, where $\wdist_{G_A'}$ is the distance in $G_A'$ as computed by Alice, and $\wdist_{G_B}$ is the distance in $G_B$ that was sent by Bob.
	
	For Bob to compute his required weighted distances, for each node $u\in C_A$, similar information is sent by Alice to Bob, that is, Alice sends to Bob the weighted distances in $G_A$ from $u$ to all nodes in $V_A$. Bob constructs the analogous graph $G_B'$ and outputs his required distance. The next paragraph formalizes this for completeness, but may be skipped by a convinced reader.
	
	Formally, Alice sends $\{\wdist_{G_A}(u,v) \mid u\in C_A, v\in V_A\}$. Bob constructs $G_B'=(V_B',E_B',w_B')$ with $V_B' = V_B\cup C_A$ and edges $E_B'=E_B\cup C \cup (C_A\times C_A)$. The edge-weight function $w_B'$ is defined by $w_B'(e)=w(e)$ for each $e\in E_B\cup C$, and $w_B'(u,v)$ for $u,v\in C_A$ is defined to be the weighted distance between $u$ and $v$ in $G_A$, as received from Alice (or $\infty$ if they are not connected in $G_A$). Bob then computes the set of all weighted distances in $G_B'$, $\{\wdist_{G_B'}(u,v) \mid u, v\in V_B'\}$. Bob assigns his output for the weighted distances in $G$ as follows. For two nodes $u,v\in V_B\cup C_A$, Bob outputs their weighted distance in $G_B'$, $\wdist_{G_B'}(u,v)$. For a node $u\in V_B'$ and a node $v\in V_A\setminus C_A$, Bob outputs $\min\{\wdist_{G_B'}(u,x)+\wdist_{G_A}(x,v)\mid x\in C_A\}$, where $\wdist_{G_B'}$ is the distance in $G_B'$ as computed by Bob, and $\wdist_{G_A}$ is the distance in $G_A$ that was sent by Alice.

~\\\textbf{Complexity.}
	Bob sends to Alice the distances from all nodes in $C_B$ to all node in $V_B$, which takes $O(\size{C_B}\size{V_B}\log n)$ bits, and similarly Alice sends
	$O(\size{C_A}\size{V_A} \log n)$ bits to Bob, for a total of $O(\size{V(C)}n\log n)$ bits.
	
	~\\\textbf{Correctness.}
	By construction, for every edge $(u,v)\in C_B\times C_B$ in $G_A'$ with weight $\wdist_{G_A'}(u,v)$, there is a corresponding shortest path $P_{u,v}$ of the same weight in $G_B$. Hence, for any path $P'=(v_0,v_1,\ldots,v_k)$ in $G_A'$ between $v_0,v_k \in V_A'$, there is a corresponding path $P_{v_0,v_k}$ of the same weight in $G$, where $P$ is obtained from $P'$ by replacing every two consecutive nodes $v_i, v_{i+1}$ in $P\cap C_B$ by the path $P_{v_i,v_{i+1}}$. Thus, $\wdist_{G_A'}(v_0,v_k)\geq \wdist_{G}(v_0,v_k)$.
	
	On the other hand,  for any shortest path $P=(v_0,v_1,\ldots,v_k)$ in $G$ connecting $v_0,v_k \in V_A'$, there is a corresponding path $P'$ of the same weight in $G_A'$, where $P'$ is obtained from $P$ by replacing any sub-path $(v_i,\ldots,v_j)$ of $P$ contained in $G_B$ and connecting $v_i,v_j\in C_B$ by the edge $(v_i,v_j)$ in $G_A'$. Thus, $\wdist_{G}(v_0,v_k)\geq \wdist_{G_A'}(v_0,v_k)$. Alice thus correctly computes the weighted distances between pairs of nodes in $V_A'$.
	
	It remains to argue about the weighted distances that Alice computes to nodes in $V_B\setminus C_B$. Any lightest path $P$ in $G$ connecting a node $u\in V_A'$ and a node$v\in V_B\setminus C_B$ must cross at least one edge of $C$ and thus must contain a node in $C_B$. Therefore, $\wdist_G(u,v) = \min\{\wdist_{G}(u,x)+\wdist_{G}(x,v)\mid x\in C_B\}$. Recall that we have shown that $\wdist_{G_A'}(u,x)=\wdist_{G}(u,x)$ for any
	$u,x\in V_A'$. The sub-path of $P$ connecting $x$ and $v$ is a shortest path between these nodes, and is contained in $G_B$,
	so $\wdist_{G_B}(x,v)=\wdist_{G}(x,v)$. Hence, the distance $\min\{\wdist_{G_A'}(u,x)+\wdist_{G_B}(x,v)\mid x\in C_B\}$ returned by Alice is indeed equal to $\wdist_G(u,v)$.
	
	The outputs of Bob are correct by the analogous arguments, completing the proof.
\end{proof}
}
\begin{proofof}{Theorem~\ref{thm:noAliceBob}}
Let $f:\set{0,1}^{K_1}\times\set{0,1}^{K_2}\to \set{0,1}^{L_1}\times\set{0,1}^{L_2}$ be a function and let $G_{x,y}$ be a family of lower bound graphs w.r.t. $f$ and the weighted APSP problem. By Lemma~\ref{thm:AliceBobCompute}, Alice and Bob can compute the weighted distances for any graph in $G_{x,y}$ by exchanging at most $O(|V(C)|n\log{n})$ bits, which is at most $O(|C|n\log{n})$ bits. Since $G_{x,y}$ is a family of lower bound graphs w.r.t.\ $f$ and weighted APSP, condition (4) gives that this number of bits is an upper bound for $CC(f)$. Therefore, when applying Theorem~\ref{thm: general lb framework APSP} to $f$ and $G_{x,y}$, the lower bound obtained for the number of rounds for computing weighted APSP is $\Omega(CC(f)/|C|\log{n})$, which is no higher than a bound of $\Omega(n)$.
\end{proofof}

~\\
\textbf{Extending to $t$ players:} We argue that generalizing the Alice-Bob framework to a shared-blackboard multi-party setting is still insufficient for providing a super-linear lower bound for weighted APSP. Suppose that we increase the number of players in the above framework to $t$ players, $P_0,\dots,P_{t-1}$, each simulating the nodes in a set $V_i$ in a partition of $V$ in a family of lower bound graphs w.r.t.\ a $t$-party function $f$ and  weighted APSP. That is, the outputs of nodes in $V_{i}$ for an algorithm $ALG$ for solving a problem $P$ in the CONGEST model, uniquely determines the output of player $P_i$ in the function $f$. The function $f$ is now a function from $\{0,1\}^{K_0}\times\cdots\times\{0,1\}^{K_{t-1}}$ to $\{0,1\}^{L_0}\times\cdots\times\{0,1\}^{L_{t-1}}$.

The communication complexity $CC(f)$ is the total number of bits written on the shared blackboard by all players. Denote by $C$ the set of cut edges, that is, the edge whose endpoints do not belong to the same $V_i$. Then, if  $ALG$ is a $R$-rounds algorithm, we have that writing $O(R|C|\log{n})$ bits on the shared blackboard suffice for computing $f$, and so $R=\Omega(CC(f)/|C|\log{n})$.

We now consider the problem $P$ to be weighted APSP. Let $f$ be a $t$-party function and let $G_{x_0,\dots,x_{t-1}}$ be a family of lower bound graphs w.r.t.\ $f$ and weighted APSP. We first have the players write all the edges in $C$ on the shared blackboard, for a total of $O(|C|\log{n})$ bits. Then, in turn, each player $P_i$ writes the weighted distances from all nodes in $V_i$ to all nodes in $V(C)\cap V_i$. This requires no more than $O(|V(C)|n\log{n})$ bits.

It is easy to verify that every player $P_i$ can now compute the weighted distances from all nodes in $V_i$ to all nodes in $V$, in a manner that is similar to that of Lemma~\ref{thm:AliceBobCompute}.

This gives an upper bound on $CC(f)$, which implies that any lower bound obtained by a reduction from $f$ is $\Omega(CC(f)/|C|\log{n})$, which is no larger than $\Omega((|V(C)|n+|C|)\log{n}/|C|\log{n})$, which is no larger than $\Omega(n)$, since $|V(C)|\leq 2|C|$.

~\\
\textbf{Remark 1:} Notice that the $t$-party simulation of the algorithm for the CONGEST model does not require a shared blackboard and can be done in the peer-to-peer multiparty setting as well, since simulating the delivery of a message does not require the message to be known globally. This raises the question of why would one consider a reduction to the CONGEST model from the stronger shared-blackboard model to begin with. Notice that our argument above for $t$ players does not translate to the peer-to-peer multiparty setting, because it assumes that the edges of the cut $C$ can be made global knowledge within writing $|C|\log{n}$ bits on the blackboard. However, what our extension above shows is that if there is a lower bound that is to be obtained using a reduction from peer-to-peer $t$-party computation, \emph{it must use a function $f$ that is strictly harder to compute in the peer-to-peer setting compared with the shared-blackboard setting}.

~\\
\textbf{Remark 2:} We suspect that a similar argument can be applied for the framework of non-fixed Alice-Bob partitions (e.g.,~\cite{SarmaHKKNPPW12}), but this requires precisely defining this framework which is not addressed in this version.

\section{Discussion}
This work provides the first super-linear lower bounds for the CONGEST model, raising a plethora of open questions. First, we showed for some specific problems, namely, computing a minimum vertex cover, a maximum independent set and a $\chi$-coloring, that they are nearly as hard as possible for the CONGEST model. However, we know that approximate solutions for some of these problems can be obtained much faster, in a polylogarithmic number of rounds or even less. A family of specific open questions is then to characterize the exact trade-off between approximation factors and round complexities for these problems.

Another specific open question is the complexity of weighted APSP, which has also been asked in previous work~\cite{Elkin04, Nanongkai14}. Our proof that the Alice-Bob framework is incapable of providing super-linear lower bounds for this problem may be viewed as providing evidence that weighted APSP can be solved much faster than is currently known. Together with the recent sub-quadratic algorithm of~\cite{Elkin17}, this brings another angle to the question: can weighted APSP be solved in linear time?

Finally, we propose a more general open question which addresses a possible classification of complexities of global problems in the CONGEST model. Some such problems have complexities of $\Theta(D)$, such as constructing a BFS tree. Others have complexities of $\tilde{\Theta}(D+\sqrt{n})$, such as finding an MST. Some problems have near-linear complexities, such as unweighted APSP. And now we know about the family of hardest problems for the CONGEST model, whose complexities are near-quadratic. Do these complexities capture all possibilities, when natural global graph problems are concerned? Or are there such problems with a complexity of, say, $\Theta(n^{1+\delta})$, for some constant $0<\delta<1$? A similar question was recently addressed in~\cite{ChangP17} for the LOCAL model, and we propose investigating the possibility that such a hierarchy exists for the CONGEST model.

~\\
\textbf{Acknowledgement:} We thank Amir Abboud, Ohad Ben Baruch, Michael Elkin, Yuval Filmus and Christoph Lenzen for useful discussions.

\renewcommand{\baselinestretch}{0.9}
\bibliographystyle{alpha}
\bibliography{bib}

\appendix
\renewcommand{\baselinestretch}{1}
\ifabs{
\section{Graph Coloring}
\label{sec:coloring}
\ColoringSection
}{}
	
\ifabs{
\section{Quadratic and Near-Quadratic Lower Bounds for Problems in P}
\label{sec:P}
\Psection
}{}

\ifabs{
\section{Finishing the Proof of Lemma~\ref{thm:AliceBobCompute}}
\label{append:AliceBobCompute}
~\\\textbf{Complexity.}
Bob sends to Alice the distances from all nodes in $C_B$ to all node in $V_B$, which takes $O(\size{C_B}\size{V_B}\log n)$ bits, and similarly Alice sends
$O(\size{C_A}\size{V_A} \log n)$ bits to Bob, for a total of $O(\size{V(C)}n\log n)$ bits.

~\\\textbf{Correctness.}
By construction, for every edge $(u,v)\in C_B\times C_B$ in $G_A'$ with weight $\wdist_{G_A'}(u,v)$, there is a corresponding shortest path $P_{u,v}$ of the same weight in $G_B$. Hence, for any path $P'=(v_0,v_1,\ldots,v_k)$ in $G_A'$ between $v_0,v_k \in V_A'$, there is a corresponding path $P_{v_0,v_k}$ of the same weight in $G$, where $P$ is obtained from $P'$ by replacing every two consecutive nodes $v_i, v_{i+1}$ in $P\cap C_B$ by the path $P_{v_i,v_{i+1}}$. Thus, $\wdist_{G_A'}(v_0,v_k)\geq \wdist_{G}(v_0,v_k)$.

On the other hand,  for any shortest path $P=(v_0,v_1,\ldots,v_k)$ in $G$ connecting $v_0,v_k \in V_A'$, there is a corresponding path $P'$ of the same weight in $G_A'$, where $P'$ is obtained from $P$ by replacing any sub-path $(v_i,\ldots,v_j)$ of $P$ contained in $G_B$ and connecting $v_i,v_j\in C_B$ by the edge $(v_i,v_j)$ in $G_A'$. Thus, $\wdist_{G}(v_0,v_k)\geq \wdist_{G_A'}(v_0,v_k)$. Alice thus correctly computes the weighted distances between pairs of nodes in $V_A'$.

It remains to argue about the weighted distances that Alice computes to nodes in $V_B\setminus C_B$. Any lightest path $P$ in $G$ connecting a node $u\in V_A'$ and a node$v\in V_B\setminus C_B$ must cross at least one edge of $C$ and thus must contain a node in $C_B$. Therefore, $\wdist_G(u,v) = \min\{\wdist_{G}(u,x)+\wdist_{G}(x,v)\mid x\in C_B\}$. Recall that we have shown that $\wdist_{G_A'}(u,x)=\wdist_{G}(u,x)$ for any
$u,x\in V_A'$. The sub-path of $P$ connecting $x$ and $v$ is a shortest path between these nodes, and is contained in $G_B$,
so $\wdist_{G_B}(x,v)=\wdist_{G}(x,v)$. Hence, the distance $\min\{\wdist_{G_A'}(u,x)+\wdist_{G_B}(x,v)\mid x\in C_B\}$ returned by Alice is indeed equal to $\wdist_G(u,v)$.

The outputs of Bob are correct by the analogous arguments, completing the proof.
}{}

\end{document}